%% file: 00Main.tex
\documentclass[lettersize,journal]{IEEEtran}

\usepackage[caption=false,font=normalsize]{subfig}

\usepackage[final]{graphicx}

\PassOptionsToPackage{shortlabels}{enumitem}
\input{preamble/preamble}

\input{preamble/letterfonts}


\usepackage{cite}
\usepackage{setspace}

\usepackage{tikz}

\usepackage{booktabs}

\usetikzlibrary{arrows.meta, positioning, calc}

\definecolor{superlinearColor}{RGB}{253, 235, 208}
\definecolor{linearColor}{RGB}{209, 242, 235}
\definecolor{boundaryRed}{RGB}{203, 67, 53}
\definecolor{textDark}{RGB}{44, 62, 80}

\graphicspath{{./figs/}}
\newcommand{\bSigma}{\boldsymbol{\Sigma}}
\newcommand{\bhSigma}{\widehat{\boldsymbol{\Sigma}}}

\newcommand{\ku}{k_\bfu}
\newcommand{\kv}{k_\bfv}

\newtheorem{theorem}{Theorem}
\newtheorem{proposition}{Proposition}
\newtheorem{lemma}{Lemma}
\newtheorem{corollary}{Corollary}
\theoremstyle{remark}

\newtheorem{definition}{Definition}

\Crefname{figure}{Fig.}{Fig.}

\begin{document}

\title{Beyond the Flat-Spike: Adaptive Sparse CCA for Decaying and Unbalanced Signals}
\author{Mengchu~Xu,~\IEEEmembership{Member,~IEEE},
	Jian~Wang,~\IEEEmembership{Member,~IEEE},
	and~Yonina~C.~Eldar,~\IEEEmembership{Fellow,~IEEE}
	\thanks{This work was supported in part by the European Research Council (ERC) through the European Union's Horizon 2020 Research and Innovation Program under Grant 101000967, in part by the Israel Science Foundation under Grant 536/22, in part by the National Key R\&D Program of China under Grant 2024YFF0505601, and in part by the National Natural Science Foundation of China under Grant 62471147 and Grant 62601183. (Corresponding author: Mengchu Xu.)}
	\thanks{Mengchu Xu is with the School of Data Science, Fudan University, Shanghai 200433, China. He was with the Faculty of Mathematics and Computer Science, Weizmann Institute of Science, Rehovot 7610001, Israel, when this work was performed (e-mail: mcxu@fudan.edu.cn; mengchu.xu@weizmann.ac.il).}
	\thanks{Jian Wang is with the School of Data Science, Fudan University, Shanghai 200433, China (e-mail: jian\_wang@fudan.edu.cn).}
	\thanks{Yonina C. Eldar is with the Faculty of Mathematics and Computer Science, Weizmann Institute of Science, Rehovot 7610001, Israel (e-mail: yonina.eldar@weizmann.ac.il), and also with the Department of Electrical and Computer Engineering, Northeastern University, Boston, MA 02115 USA (e-mail: y.eldar@northeastern.edu).}
}
\date{}

\maketitle

\begin{abstract}
	Sparse Canonical Correlation Analysis (SCCA) is a fundamental statistical tool for identifying linear relationships in high-dimensional, multi-view data. While minimax theory establishes an optimal sample complexity scaling additively with the sparsity levels of the canonical vectors, computationally efficient algorithms typically suffer from a suboptimal multiplicative dependence. This computational-statistical gap is intrinsically tied to worst-case ``flat'' signal assumptions. In practice, however, multi-view signals frequently exhibit structured energy concentration, such as a power-law decay. To exploit this structural concentration and bypass the worst-case bottleneck, we propose Bilateral Spectral Energy Pursuit (Bi-SEP). Operating directly on the cross-covariance matrix, Bi-SEP is a stagewise adaptive algorithm that utilizes a proxy refinement step to dynamically track and capture cross-view signal energy. Theoretically, we establish a profile-adaptive sample complexity bound governed by the coupled energy profiles of the two views. Notably, under power-law decay models, we reveal a synergistic phase transition: the sufficient sample complexity becomes linear in sparsity, up to logarithmic factors, provided that the aggregate decay rate of the two views is sufficiently large. This result demonstrates that a highly concentrated signal in one view allows the model to accommodate a completely flat signal in its partner. Numerical experiments validate our theoretical findings, illustrating the advantages of Bi-SEP in structured, non-flat signal regimes.
\end{abstract}


\input{01Introduction.tex}

\input{02Algorithm.tex}

\input{03Results.tex}

\input{04Proofs.tex}

\input{05Discussion.tex}
\input{06Experiments.tex}
\input{07Conclusion.tex}

\input{08Appendices.tex}

\input{00Main.bbl}
\end{document}

%% file: preamble/preamble.tex
\usepackage{amsmath,amsfonts,amssymb,mathtools}
\usepackage{empheq}
\usepackage{booktabs}
\usepackage{microtype}
\usepackage{xfrac}
\usepackage[makeroom]{cancel}
\usepackage{varwidth}
\usepackage{etoolbox}
\usepackage{nameref}
\usepackage{multicol,array}
\usepackage{comment}
\usepackage{import}
\usepackage{xifthen}
\usepackage{pdfpages}
\usepackage{transparent}
\usepackage{chngcntr}
\usepackage{tikz}
\usepackage{setspace}

\usepackage{algorithm}
\usepackage{algpseudocode}

\makeatletter
\@ifclassloaded{beamer}{
	\typeout{>>> Class [beamer] detected: Applying slide preamble configuration...}
}{
	\typeout{>>> Regular Class detected: Applying document preamble configuration...}
	\usepackage{amsthm}
	\usepackage{enumitem}
	\usepackage{titletoc}
	\usepackage{theoremref} %

	\usepackage[cmintegrals]{newtxmath}
}
\makeatother

\usepackage{xspace} %

\usepackage{hyperref}
\usepackage{bookmark}
\usepackage{cleveref}

\hypersetup{
	hidelinks,
	breaklinks=true,
	bookmarksnumbered=true,
	bookmarksopen=true,
	pdfdisplaydoctitle=true,
	pdfstartview={FitH},
	pdfpagemode=UseOutlines,
	unicode=true
}

%% file: preamble/letterfonts.tex
\newcommand{\bbE}{\mathbb{E}}

	\newcommand{\bbP}{\mathbb{P}}
	\newcommand{\bbR}{\mathbb{R}}

\newcommand{\bfA}{\mathbf{A}}	
	
\newcommand{\bfE}{\mathbf{E}}	
	
\newcommand{\bfI}{\mathbf{I}}

\newcommand{\bfW}{\mathbf{W}}	
	\newcommand{\bfZ}{\mathbf{Z}}

\newcommand{\bfa}{\mathbf{a}}	\newcommand{\bfb}{\mathbf{b}}
	
\newcommand{\bfe}{\mathbf{e}}	
	\newcommand{\bfh}{\mathbf{h}}

	\newcommand{\bfr}{\mathbf{r}}
	
\newcommand{\bfu}{\mathbf{u}}	\newcommand{\bfv}{\mathbf{v}}
	\newcommand{\bfx}{\mathbf{x}}
\newcommand{\bfy}{\mathbf{y}}

\newcommand{\mcE}{\mathcal{E}}

	\newcommand{\mcN}{\mathcal{N}}

\newcommand{\xmax}{x_{\text{max}}}

%% file: 01Introduction.tex
\section{Introduction}
\label{sec:introduction}

Canonical Correlation Analysis (CCA), initially introduced by Hotelling~\cite{Hotelling1936Relations}, remains a cornerstone of multivariate statistics for identifying linear relationships between two sets of variables. In the modern era of data science, applications spanning genomics, neuroimaging, and financial econometrics increasingly demand the analysis of high-dimensional multi-view data~\cite{parkhomenko2009sparse,Avants2014Sparse,Wilms2015Sparse}. Furthermore, within the broader signal processing community, CCA has been extensively leveraged for array signal processing, blind source separation, and multi-sensor fusion \cite{scharf2000canonical, via2005canonical, li2008cca, hardoon2004canonical}, where the number of features significantly exceeds the sample size. In these high-dimensional, small-sample regimes, classical CCA encounters fundamental limitations: sample covariance matrices become singular or highly ill-conditioned, and the empirical estimators suffer from statistical inconsistency~\cite{Johnstone2001Distribution, ledoit2004well, bickel2008covariance, vershynin2012introduction}. Furthermore, the resulting canonical vectors, which are typically dense, lack the interpretability required for scientific discovery.

To address these challenges, Sparse CCA (SCCA) has emerged as the standard paradigm, enforcing sparsity constraints to recover interpretable leading canonical pairs~\cite{Zou2006Sparse,Hardoon2011SCCA}. The theoretical landscape of SCCA was formalized by Gao, Ma, and Zhou~\cite{gao2017sparse}, who established that under the ideal ``exact sparsity'' setting, the minimax optimal sample complexity scales as $m \asymp (k_\bfu + k_\bfv) \log n$. However, they also proved a computational lower bound, highlighting a fundamental statistical-to-computational gap~\cite{gao2017sparse,brennan2020reducibility}: polynomial-time algorithms face an unavoidable bottleneck of $m \gtrsim k_\bfu k_\bfv \log n$, established via reduction from the Planted Clique problem.

This computational barrier is primarily driven by a fundamental algorithmic hurdle inherent to SCCA: the bilateral deadlock. Because the true signal is encoded exclusively through the cross-correlation $\rho \bfu \bfv^\top$, accurately estimating the support of $\bfu$ requires an estimate of $\bfv$ that is already sufficiently aligned, and vice versa. While computationally efficient iterative approaches can achieve optimal linear convergence rates, such as Penalized Matrix Decomposition (PMD)~\cite{witten2009penalized}, Truncated Power Method (TPower) variants~\cite{Yuan2013TPower, chen2012structured, asteris2016simple}, and majorization-minimization algorithms~\cite{archambeau2008sparse, sriperumbudur2011majorization}, they often bypass this deadlock in theory by assuming an initialization of sufficient quality~\cite{jain2013low, wang2014tighten, balakrishnan2017statistical}. In practice, obtaining such an initial estimate from scratch typically requires sample complexity of order $k_\bfu k_\bfv\log n$ in the worst case~\cite{gao2017sparse}.

Notably, this computational hardness is intrinsically tied to worst-case ``flat'' signals, where all active canonical weights are uniformly weak. In modern multi-view applications, however, signals rarely exhibit such an adversarial structure. Instead, consistent with established compressibility priors in natural images and audio~\cite{field1987relations, mallat1999wavelet}, canonical components typically exhibit a coupled energy decay (e.g., following a power-law distribution~\cite{donoho2006compressed, candes2008introduction}), where a small subset of features captures the dominant cross-view correlation. Exploiting this natural structural concentration has the potential to overcome the multiplicative bottleneck: strong prominent coordinates can act as reliable anchors, facilitating the recovery of the remaining weaker entries.

Yet, developing a profile-adaptive framework remains challenging due to the reliance on profile-blind truncation. Because standard approaches enforce rigid global thresholds or fixed-cardinality constraints (often seen in iterative hard thresholding paradigms~\cite{blumensath2009iterative, Ma2013SPCAIT, cai2011orthogonal}), they process all non-zero coefficients uniformly. Consequently, an algorithm cannot lower its threshold to capture trailing, weaker coefficients without simultaneously introducing excessive ambient cross-covariance noise. To ensure theoretical tractability under this algorithmic rigidity, existing analyses often resort to the worst-case flat-signal assumption, leaving the synergistic benefits of decaying profiles largely unquantified.

To bypass the limitations of premature uniform truncation, we build on the Spectral Energy Pursuit (SEP) framework for SPCA~\cite{xu2025spca}, which uses a signal energy structure function and progressively expands and reselects the active support. However, extending this pursuit mechanism from the symmetric, rank-one estimation of SPCA to the asymmetric structure of SCCA is nontrivial due to the inherent bilateral deadlock discussed earlier. In SPCA, the pursuit mechanism is self-reinforcing: once a coordinate begins to align with the signal, its contribution naturally amplifies in subsequent iterations, allowing dominant components to progressively emerge. Typically, this process can be initiated from coordinates associated with large diagonal variances. In contrast, SCCA operates on cross-view interactions. Dictated by the bilateral deadlock discussed earlier, concentrated energy in one view can only be revealed if the other view already provides a sufficiently aligned filtering direction. This brings the challenge full circle back to the initialization bottleneck: without a reliable starting anchor, an adaptive energy pursuit cannot even commence. A consolidated comparison of the ingredients inherited from SEP and the developments specific to SCCA is provided in \Cref{sec:discussion-sep}.

To explicitly break this dependency loop and bridge the gap between theory and practice from scratch, we propose \textit{Bilateral Spectral Energy Pursuit (Bi-SEP)}. By executing a principled zero-step proxy generation that safely decouples the initial cross-view dependency, alongside a stagewise expansion that tracks signal concentration, Bi-SEP naturally adapts to arbitrary signal decay profiles and asymmetric sparsity budgets ($k_\bfu \neq k_\bfv$). Our main theoretical contributions directly address the aforementioned bottlenecks:

\begin{enumerate}
	\item \textbf{Profile-Adaptive Sample Complexity and Synergistic Phase Transition:}
	      We provide a rigorous finite-sample analysis demonstrating that Bi-SEP bypasses the worst-case flat-signal bottleneck. By characterizing the dual signal energy structure functions (\Cref{def:structure_function}), we derive sample complexity guarantees governed by the coupled energy profiles of the two views. In particular, we show that near-linear sample complexity, up to logarithmic factors, is achievable whenever the two views exhibit sufficient joint concentration of signal energy. As an illustrative example, under power-law decay models (\Cref{def:power_law}), this condition corresponds to $\alpha_\bfu+\alpha_\bfv\ge 1$; see \Cref{cor:scca_power_law_transition}. This reveals a key conceptual insight: a highly concentrated signal in one view can effectively compensate for a structurally flat signal in the other, and thus stands in contrast to the classical minimax theory for SCCA, where the sample complexity is governed by worst-case flat signals.

	\item \textbf{Provable Decoupling of the Initialization Bottleneck:}
	      We theoretically justify a novel zero-step proxy generation mechanism that bridges the critical gap between initialization and stagewise iteration. We show that this carefully designed step explicitly converts a joint-product initialization bound into strictly valid, element-wise energy lower bounds. This decoupling establishes the fundamental base case for the algorithmic induction, ensuring provable convergence from a locally maximal entry without requiring artificially strong initialization assumptions.
\end{enumerate}

Throughout this paper, boldface lowercase letters (e.g., $\bfx$) and uppercase letters (e.g., $\bfW$) are used to denote vectors and matrices, respectively. For a positive integer $n$, we define $[n] := \{1, 2, \dots, n\}$. For a subset of indices $S \subset [n]$, $|S|$ is its cardinality. $\bfx_S \in \mathbb{R}^{|S|}$ is the subvector of $\bfx$ restricted to the indices in $S$, and $\bfW_{S_1, S_2}$ is the submatrix of $\bfW$ formed by the rows in $S_1$ and columns in $S_2$. The $\ell_2$-norm and $\ell_0$-norm of a vector are written as $\|\bfx\|_2$ and $\|\bfx\|_0$, respectively. For a vector $\bfx \in \mathbb{R}^n$, $x_i$ is its $i$-th coordinate, and $x_{(i)}$ denotes its $i$-th largest entry in absolute value, such that $|x_{(1)}| \ge |x_{(2)}| \ge \dots \ge |x_{(n)}|$. For a matrix $\bfW$, $\|\bfW\|_2$ represents its spectral norm. Additionally, $C,c>0$ denote numerical constants independent of the ambient dimension and sparsity levels; a generic $C$ may change from line to line unless it is indexed. Finally, we use standard asymptotic notation: $a \gtrsim b$ (and equivalently $b \lesssim a$) means $a \ge C b$ for some universal constant $C > 0$, and $a \asymp b$ indicates that both $a \gtrsim b$ and $a \lesssim b$ hold.

The remainder of this paper is organized as follows. Section \ref{sec:algorithm} introduces the whitened spike model, necessary notations, and the Bi-SEP algorithm. Section \ref{sec:results} presents the main theoretical recovery guarantees and the phase transition corollary. Section \ref{sec:proofs} outlines the core proofs of the theoretical results, and Section \ref{sec:discussion} offers a deeper discussion on the decoupling mechanisms and inherent error bounds coupling. Section \ref{sec:experiments} presents synthetic experiments to validate the theoretical claims. Finally, Section \ref{sec:conclusion} concludes the paper.

%% file: 02Algorithm.tex
\section{Algorithm}
\label{sec:algorithm}
\subsection{Problem Setup}

SCCA seeks sparse linear combinations of two sets of variables that maximize their cross-correlation. Let $\bfx, \bfy \in \mathbb{R}^n$ be two random vectors with cross-covariance $\bSigma_{xy}$ and marginal covariances $\bSigma_{xx}, \bSigma_{yy}$. For simplicity of exposition, we use a common ambient dimension $n$ here, though the analysis extends verbatim to the asymmetric case $\bfx\in\mathbb{R}^{n_1}$ and $\bfy\in\mathbb{R}^{n_2}$ with minor notational changes. SCCA seeks canonical weight vectors $\bfa$ and $\bfb$ by solving the following constrained optimization problem:
\begin{equation}\label{eq:scca-prob}
	\begin{aligned}
		\max_{\bfa, \bfb} \quad & \bfa^\top \bSigma_{xy} \bfb                                         \\
		\text{subject to} \quad & \bfa^\top \bSigma_{xx} \bfa = 1, \ \bfb^\top \bSigma_{yy} \bfb = 1, \\
		                        & \|\bfa\|_0 \le k_\bfa, \ \|\bfb\|_0 \le k_\bfb,
	\end{aligned}
\end{equation}
where $k_\bfa$ and $k_\bfb$ denote the sparsity budgets for each view.

For theoretical analysis in the high-dimensional regime, we assume $m$ independent and identically distributed samples $\{(\bfx_i, \bfy_i)\}_{i=1}^m$ drawn from a zero-mean joint multivariate Gaussian distribution. We adopt the whitened spiked covariance model \cite{Johnstone2001Distribution, gao2017sparse}, where the marginal covariances are identity matrices, i.e., $\bSigma_{xx} = \bSigma_{yy} = \bfI_n$. This reduces the variance constraints to unit $\ell_2$ norm constraints, i.e., $\|\bfa\|_2 = \|\bfb\|_2 = 1$. Additionally, the population cross-covariance matrix is assumed to have the rank-1 structure
\begin{equation}
	\bSigma_{xy} = \mathbb{E}[\bfx \bfy^\top] = \rho \bfu \bfv^\top,
\end{equation}
where $\rho \in (0, 1)$ is the canonical correlation. The underlying signals $\bfu, \bfv \in \mathbb{R}^n$ represent the true sparse canonical weight vectors, satisfying $\|\bfu\|_2 = \|\bfv\|_2 = 1$ with sparsity levels $\|\bfu\|_0 \le k_\bfu$ and $\|\bfv\|_0 \le k_\bfv$.

Given the sample cross-covariance matrix $\bhSigma_{xy} = \frac{1}{m} \sum_{i=1}^m \bfx_i \bfy_i^\top$, our objective is to estimate the true canonical pair $(\bfu, \bfv)$. This leads to the empirical counterpart of the whitened SCCA problem:
\begin{equation}\label{eq:scca-whitened}
	\begin{aligned}
		\max_{\hat{\bfu}, \hat{\bfv}} \quad & \hat{\bfu}^\top \bhSigma_{xy} \hat{\bfv}                                                                    \\
		\text{subject to} \quad             & \|\hat{\bfu}\|_2 = 1, \ \|\hat{\bfv}\|_2 = 1, \ \|\hat{\bfu}\|_0 \le k_\bfu, \ \|\hat{\bfv}\|_0 \le k_\bfv,
	\end{aligned}
\end{equation}
which constitutes the core estimation task studied in this paper. Due to the presence of the nonconvex $\ell_0$ norm constraints, directly solving \eqref{eq:scca-whitened} is computationally intractable. This fundamental challenge motivates the need for efficient approximation algorithms. Before introducing our proposed methodology, we briefly review existing approaches and their algorithmic limitations in addressing this specific formulation.

\subsection{Prior Art and Algorithmic Limitations}

To contextualize the mechanics of the proposed algorithm, we briefly review standard methodologies for SCCA and their structural constraints. While other perspectives such as probabilistic modeling exist~\cite{archambeau2008sparse, sriperumbudur2011majorization}, computationally efficient algorithmic approaches broadly fall into two families, neither of which explicitly exploits signals with nonuniform energy distributions.

The first family relies on convex relaxations, notably Penalized Matrix Decomposition (PMD) \cite{witten2009penalized} and its structurally penalized variants \cite{chen2012structured}. PMD enforces sparsity by solving the penalized optimization problem:
\begin{equation}
	\begin{aligned}
		\max_{\bfu, \bfv} \quad & \bfu^\top \bhSigma_{xy} \bfv                                                                      \\
		\text{subject to} \quad & \|\bfu\|_2^2 \le 1, \|\bfv\|_2^2 \le 1, \mathcal{P}_1(\bfu) \le c_1, \mathcal{P}_2(\bfv) \le c_2,
	\end{aligned}
\end{equation}
where $\mathcal{P}(\cdot)$ is typically an $\ell_1$ norm or a structurally induced penalty. While PMD provides a computationally tractable framework, its reliance on global penalty parameters $c_1$ and $c_2$ enforces a uniform shrinkage pattern. In regimes where the underlying signal exhibits a power-law decay, applying a single global threshold often fails to capture trailing coefficients without simultaneously introducing ambient cross-covariance noise into the estimate. Furthermore, while the statistical properties of the globally optimal relaxed objective have been studied extensively, the alternating PMD algorithm solves a biconvex problem and generally guarantees convergence only to a local stationary point. Consequently, it lacks sharp, nonasymptotic recovery guarantees capable of adapting to varying signal energy profiles and signal-to-noise ratios (SNR) across coordinates.

The second family, which serves as the primary baseline in our experiments, employs nonconvex alternating maximization equipped with hard cardinality constraints. A canonical example is the Truncated Power Method (TPower) and its alternating extensions to SCCA \cite{Yuan2013TPower,asteris2016simple}. These algorithms seek to recover the leading canonical pair by alternatingly applying a hard thresholding operator $\mathcal{H}_k(\cdot)$, which retains only the $k$ largest entries in absolute value:
\begin{equation}
	\hat{\bfu}^{(t+1)} \propto \mathcal{H}_{\ku}(\bhSigma_{xy} \hat{\bfv}^{(t)}), \quad
	\hat{\bfv}^{(t+1)} \propto \mathcal{H}_{\kv}(\bhSigma_{xy}^\top \hat{\bfu}^{(t+1)}),
\end{equation}
where $\mathcal{H}_k(\cdot)$ is followed by $\ell_2$-normalization.
While computationally efficient, this standard iteration template is constrained by two structural limitations. First, it exhibits a strong initialization dependency. Because the cross-view signal is encoded exclusively through the correlation matrix $\rho \bfu\bfv^\top$, the reliability of the update for $\hat{\bfu}^{(t+1)}$ strictly depends on the alignment quality of $\hat{\bfv}^{(t)}$. Under standard initialization strategies used to establish theoretical guarantees (e.g., using the largest magnitude entry of $\bhSigma_{xy}$), the initial alignment is generally insufficient to suppress ambient noise. Consequently, without stringent initialization conditions, alternating maximization frameworks incur a suboptimal sample complexity scaling as $m \gtrsim \ku \kv \log n$~\cite{gao2017sparse,brennan2020reducibility}.

Second, the algorithm encounters a uniform truncation limitation. In the early iterations of the algorithm, the underlying estimates are inherently inaccurate. By enforcing constraints with fixed cardinality ($\ku$ and $\kv$) immediately from the start, the algorithm is forced to select noisy coordinates merely to fulfill the sparsity quota, long before the dominant signal components are securely aligned. For signals characterized by a decaying energy profile, this premature uniform truncation obscures the statistical advantage provided by highly concentrated signal entries.

\subsection{Bilateral Spectral Energy Pursuit (Bi-SEP)}

To address the initialization dependency and the premature uniform truncation observed in alternating maximization schemes, we introduce Bilateral Spectral Energy Pursuit (Bi-SEP), detailed in Algorithm \ref{alg:bi_sep}. Bi-SEP replaces the iteration, which fixes the cardinality from the outset, with a stagewise adaptive framework. It dynamically constructs the support sets of $\bfu$ and $\bfv$ by progressively capturing the dominant marginal energies from the cross-covariance matrix $\bhSigma_{xy}$.

The algorithm operates through two primary mechanisms.
\begin{enumerate}
	\item \textbf{Decoupled Initialization ($t=0$):}
	      Initialization begins from the maximal empirical entry $(i_0,j_0)=\arg\max_{i,j}|(\bhSigma_{xy})_{ij}|$. Although the analysis of this choice yields a lower bound on the joint signal product $|u_{i_0}v_{j_0}|$, the selected indices need not correspond to the dominant marginal coordinates of the true canonical vectors because the bound controls only their product. To decouple this joint dependence, Bi-SEP restricts the initial support size to one and immediately performs a restricted estimation on the corresponding $1\times1$ submatrix (lines 5--9 in Algorithm~\ref{alg:bi_sep}). The resulting rank-1 estimate generates proxy vectors $\bfr_\bfu$ and $\bfr_\bfv$, from which the top-1 coordinates are reselected. This procedure converts the joint-product guarantee into element-wise energy bounds, providing a statistically reliable base configuration before support expansion.

	\item \textbf{Stagewise Expansion ($t>0$):}
	      At iteration $t$, Bi-SEP performs a restricted singular value decomposition (SVD) on the submatrix indexed by the current active supports $S_\bfu^{(t)}$ and $S_\bfv^{(t)}$. The resulting normalized estimates produce proxy vectors $(\bfr_\bfu, \bfr_\bfv)$ over the full ambient dimension. The active supports are then globally reselected by retaining the top-$\min(t+1,k)$ coordinates (with $k \in \{k_\bfu, k_\bfv\}$) from the proxy vectors. This stagewise growth gradually relaxes the sparsity constraint, allowing structurally dominant coordinates to be incorporated before weaker components enter the support. By securing the anchors with high energy first, Bi-SEP naturally exploits the intrinsic energy decay profile of sparse canonical vectors. The highly concentrated entries establish strong cross-view alignment that is less affected by noise, which subsequently guides the recovery of weaker trailing coefficients. Because the supports are recomputed globally at every iteration, early selection errors caused by weak cross-view alignment can be automatically corrected as stronger signal components emerge. When $k_\bfu \neq k_\bfv$, the two supports reach their respective sparsity budgets at different iterations, giving rise to a staggered saturation-and-hold mechanism. That is, after one view reaches its prescribed budget, its support cardinality remains fixed, whereas its indices continue to be globally reselected as the partner support expands.
\end{enumerate}

\begin{algorithm}[t]
	\caption{Bilateral Spectral Energy Pursuit (Bi-SEP)}
	\label{alg:bi_sep}
	\begin{algorithmic}[1]
		\Require Sample cross-covariance matrix $\bhSigma_{xy}$, sparsity budgets $k_\bfu, k_\bfv$.
		\State \textbf{Initialize:} Select $(i_0, j_0) = \arg\max_{i,j} |(\bhSigma_{xy})_{ij}|$.
		\State Set initial supports $S_\bfu^{(0)} = \{i_0\}$, $S_\bfv^{(0)} = \{j_0\}$.
		\State Let $k_{\max} = \max(k_\bfu, k_\bfv)$.

		\For{$t = 0$ to $k_{\max}-1$} \Comment{decoupling when $t=0$}
		\State \textbf{1. Restricted Estimation:}
		\State Compute the leading singular pair $(\hat{\bfu}^{(t)}, \hat{\bfv}^{(t)})$ of the submatrix $(\bhSigma_{xy})_{S_\bfu^{(t)}, S_\bfv^{(t)}}$. Implicitly pad estimates with zeros to dimension $n$. \Comment{Unit norm}

		\State \textbf{2. Proxy Generation:}
		\State $\bfr_\bfu \leftarrow \bhSigma_{xy} \hat{\bfv}^{(t)}$
		\State $\bfr_\bfv \leftarrow \bhSigma_{xy}^\top \hat{\bfu}^{(t)}$

		\State \textbf{3. Saturate and Hold Selection:}
		\State $S_\bfu^{(t+1)} \leftarrow$ indices of top-$\min(t{+}1, k_\bfu)$ entries of $|\bfr_\bfu|$.
		\State $S_\bfv^{(t+1)} \leftarrow$ indices of top-$\min(t{+}1, k_\bfv)$ entries of $|\bfr_\bfv|$.
		\EndFor

		\State \textbf{Final Estimation:}
		\State Compute the leading singular pair $(\hat{\bfu}^{(k_{\max})}, \hat{\bfv}^{(k_{\max})})$ of the submatrix $(\bhSigma_{xy})_{S_\bfu^{(k_{\max})}, S_\bfv^{(k_{\max})}}$. \Comment{Unit norm}

		\State \textbf{Output:} Final estimators $\hat{\bfu}^{(k_{\max})}, \hat{\bfv}^{(k_{\max})}$ implicitly padded with zeros to $\mathbb{R}^n$.
	\end{algorithmic}
\end{algorithm}

%% file: 03Results.tex
\section{Results}
\label{sec:results}

In this section, we establish the theoretical sample complexity and recovery guarantees for the Bi-SEP algorithm. We demonstrate that Bi-SEP achieves a profile-adaptive sample complexity that bypasses standard worst-case cardinality bounds by explicitly leveraging the intrinsic energy concentration of the signals.

\subsection{Signal Energy Parameterization}

Traditional sample complexity bounds for SCCA depend strictly on the active support cardinalities $k_\bfu$ and $k_\bfv$. To mathematically quantify the nonuniform energy distribution exploited by Bi-SEP, we adapt the continuous measure of signal concentration recently introduced for SPCA in \cite{xu2025spca}.

\begin{definition}[Dual Signal Energy Structure Functions]\label{def:structure_function}
	Given unit sparse vectors $\bfu, \bfv \in \mathbb{R}^n$, we define their respective structure functions $s_\bfu(p)$ and $s_\bfv(q)$ as
	\begin{equation}\label{eq:structure_function}
		s_\bfu(p) := \Big(\sum_{i=1}^{p} u_{(i)}^2\Big)^{-1}, \quad s_\bfv(q) := \Big(\sum_{j=1}^{q} v_{(j)}^2\Big)^{-1}
	\end{equation}
	for all $p,q \in [n]$. Note that since $\bfu$ and $\bfv$ are strictly sparse and possess unit $\ell_2$-norm, their energy accumulation saturates at their true sparsities. Consequently, $s_\bfu(p) = 1$ for all $p \ge k_\bfu$, and symmetrically $s_\bfv(q) = 1$ for all $q \ge k_\bfv$.
\end{definition}

The structure function $s_\bfu(p)$ quantifies the inverse cumulative energy within the top $p$ coordinates of $\bfu$. For a uniform (worst-case) sparse vector where $|u_{(i)}| = 1/\sqrt{k_\bfu}$, the function scales as $s_\bfu(p) = k_\bfu / p$. Conversely, for highly concentrated signals, $s_\bfu(p)$ approaches $1$ rapidly even for small $p$. This parameterization facilitates theoretical guarantees that dynamically reflect the true energy decay profiles of the underlying views. A canonical example of such concentration is the power-law decay profile.

\begin{definition}[Power-law Decay Profile] \label{def:power_law}
	A $k$-sparse unit vector $\bfx \in \mathbb{R}^n$ is said to follow a \textit{power-law decay profile} with rate $\alpha \ge 0$ if its nonzero entries, sorted by magnitude, satisfy
	\begin{equation}
		x_{(i)}^2 \propto i^{-\alpha}, \quad i=1,\dots,k,
	\end{equation}
	normalized such that $\|\bfx\|_2 = 1$. The parameter $\alpha$ controls the signal concentration: $\alpha=0$ corresponds to a flat (worst-case) signal, while $\alpha > 1$ denotes strong concentration.
\end{definition}

	Throughout the paper, we measure direction estimation error by the sine loss. For any nonzero vectors \(\hat{\bfx},\bfx\in\mathbb{R}^n\), we define
	\begin{equation}\label{eq:sine_loss}
		\sin\angle(\hat{\bfx},\bfx)
		:= \sqrt{1-\frac{(\hat{\bfx}^{\top}\bfx)^2}{\|\hat{\bfx}\|_2^2\|\bfx\|_2^2}} .
	\end{equation}
	This loss is invariant to the unavoidable sign ambiguity of one-dimensional canonical directions.

\subsection{Main Recovery Guarantees}
The proposed Bi-SEP algorithm operates via singular value decompositions on submatrices of $\bhSigma_{xy}$. Consequently, the theoretical performance is primarily determined by the spectral properties of the cross-term noise. We define the effective noise matrix $\bfW$ as the deviation from the population cross-covariance:
\begin{equation}\label{eq:effective_noise}
	\bfW := \bhSigma_{xy} - \rho \bfu \bfv^\top.
\end{equation}

The adaptive nature of Bi-SEP implies that its sequence of estimated supports depends on the data. Bounding the estimation error therefore requires uniform control over $\bfW$ across all possible sparse rectangular submatrices.
We formalize this via the uniform noise event:
\begin{equation}\label{eq:union_event}
	\mcE(t) := \hspace{-1mm}\bigcap_{\substack{|S_1|\le k_\bfu \\ |S_2|\le k_\bfv}} \left\{ \|\bfW_{S_1,S_2}\|_2 \le C \sqrt{\frac{(|S_1|+|S_2|)\log n+t}{m}} \right\}.
\end{equation}
Let $\mcE$ denote the event $\mcE(t)$ with $t=\log n$. As formally stated in \Cref{prop:union_event} (\Cref{sec:proofs}), this uniform spectral bound is guaranteed to hold with probability at least $1-2n^{-c}$, provided that the sample size scales appropriately with the block dimensions.

Conditioned on the event $\mcE$, \Cref{thm:recovery_guarantee} presents the primary recovery guarantee for Bi-SEP. The sample complexity is no longer bottlenecked by the terminal sparsities $k_\bfu$ and $k_\bfv$, but is intrinsically governed by the coupled energy profiles of the underlying canonical vectors.
\begin{theorem}[Profile-adaptive sample complexity for direction estimation] \label{thm:recovery_guarantee}
	Condition on the high-probability event \(\mcE\). Let $k_{\max} = \max(k_\bfu, k_\bfv)$. For any iteration step $t$, define the effective support sizes as $t_\bfu = \min(t, k_\bfu)$ and $t_\bfv = \min(t, k_\bfv)$ and fix an energy capture fraction \(\gamma\in(0,1)\). If
	\begin{equation}\label{eq:m-bound-precise}
		m \ \ge\ C_1 \max_{1\le t \le k_{\max}} \underbrace{(t_\bfu + t_\bfv)}_{\text{effective noise}} \underbrace{s_{\bfu}(t_\bfu)s_{\bfv}(t_\bfv)}_{\text{coupled profile}} \log n,
	\end{equation}
	then the estimates \(\hat\bfu^{(k_{\max})}\) and \(\hat \bfv^{(k_{\max})}\) produced by \Cref{alg:bi_sep} satisfy
	\begin{equation}\label{eq:u-error}
		\sin\angle(\hat \bfu^{(k_{\max})},\bfu)\le\underbrace{\sqrt{1-\gamma}}_{\text{approximation error}}
		+ ~
		\underbrace{C_2\sqrt{\frac{(\ku+\kv)\log n}{m}}}_{\text{statistical error}},
	\end{equation}
	and
	\begin{equation}\label{eq:v-error}
		\sin\angle(\hat \bfv^{(k_{\max})},\bfv) \le \underbrace{\sqrt{1-\gamma}}_{\text{approximation error}}
		+ ~
		\underbrace{C_2\sqrt{\frac{(\ku+\kv)\log n}{m}}}_{\text{statistical error}},
	\end{equation}
	where $C_1 = C \frac{(1+\rho)^2}{\rho^2 \gamma^2(1-\sqrt{\gamma})^2}$ and $C_2 = C \frac{(1+\rho)}{\rho \gamma}$.
\end{theorem}

\Cref{thm:recovery_guarantee} highlights three distinct factors governing recovery: the canonical correlation strength $\rho$, the stagewise energy-capture level $\gamma$, and, most importantly for profile adaptivity, the trajectory-dependent interaction between the effective noise dimension and the two signal energy profiles.

The correlation level \(\rho\) enters the sufficient sample-size prefactor through $(1+\rho)^2/\rho^2$ and the statistical-error prefactor through $(1+\rho)/\rho$. Hence, as $\rho\downarrow0$, the required sample size increases with leading dependence of order $\rho^{-2}$, while the error prefactor scales as $\rho^{-1}$. The corresponding finite-sample sensitivity is examined in \Cref{sec:exp-rho}.

The parameter $\gamma \in (0,1)$ is a proof parameter controlling the required energy capture at each stage. This produces the approximation term $\sqrt{1-\gamma}$ in \eqref{eq:u-error}--\eqref{eq:v-error}: a larger $\gamma$ reduces this term but inflates the theorem constants. For order-wise statements, $\gamma$ is treated as a fixed numerical constant, such as $1/2$. The nonvanishing approximation term can be removed by the one-step post-refinement developed in \Cref{sec:post-refinement}.

With the roles of $\rho$ and $\gamma$ clarified, we now focus on the profile-dependent sample-complexity order. The key feature of \eqref{eq:m-bound-precise} is that the trajectory bottleneck separates into the effective noise dimension $(t_\bfu+t_\bfv)$ and the coupled profile factor $s_\bfu(t_\bfu)s_\bfv(t_\bfv)$.

The explicit $(t_\bfu+t_\bfv)$ noise factor naturally prevents runaway noise accumulation in unbalanced regimes ($k_\bfu\neq k_\bfv$): once one support reaches its prescribed size, its contribution to the noise dimension no longer grows as the partner support continues to expand. At the level of asymptotic scaling, however, $t_\bfu+t_\bfv\asymp t$ for $1\le t\le k_{\max}$, so the sample-complexity order can be summarized as
\begin{equation}\label{eq:m-bound-asymp}
	m \gtrsim \max_{1 \le t \le k_{\max}} t s_\bfu(t) s_\bfv(t) \log n.
\end{equation}

This simplified expression makes the role of signal concentration transparent. If both structure functions remain of constant order, the sufficient complexity is linear in the effective support size, recovering the additive scaling $(k_\bfu+k_\bfv)\log n$ up to constants. More generally, the product $s_\bfu(t)s_\bfv(t)$ shows that concentration in the two views acts jointly: stronger concentration in one view can compensate for weaker concentration in the other. We make this tradeoff explicit next through the balanced power-law model.

\subsection{Power-Law Phase Transition}

\begin{corollary}[Phase Transition under Balanced Profiles]\label{cor:scca_power_law_transition}
	Consider the balanced sparsity regime $k_\bfu = k_\bfv = k$. Assume $\bfu$ and $\bfv$ follow power-law decay profiles with rates $\alpha_\bfu, \alpha_\bfv \ge 0$. The sample complexity exhibits a phase transition governed directly by their aggregate decay rate $\alpha_{\bfu\bfv} = \alpha_\bfu + \alpha_\bfv$:
	\begin{equation}\label{eq:phase_transition}
		m \gtrsim k^{\tau} \log n, \quad \text{where } \tau = \max(1, 2 - \alpha_{\bfu\bfv}).
	\end{equation}
	An additional $\log k$ factor arises in the boundary case $\{\alpha_\bfu, \alpha_\bfv\} = \{0, 1\}$, yielding $m \gtrsim k \log k \log n$.
\end{corollary}

\input{figs/complexity_phase_transit.tex}
As illustrated in \Cref{fig:complexity_phase_transit}, apart from the logarithmic correction at the critical boundary, \Cref{cor:scca_power_law_transition} yields two main scaling regimes:
\begin{itemize}
	\item \textbf{Super-linear Regime ($\alpha_{\bfu\bfv} < 1$):} When the signals lack sufficient concentration, the sample complexity scales super-linearly as $k^{2-\alpha_{\bfu\bfv}}$. In the worst-case flat regime ($\alpha_\bfu = \alpha_\bfv = 0$), this recovers the standard quadratic scaling $k^2 \log n$.

	\item \textbf{Linear Regime ($\alpha_{\bfu\bfv} \ge 1$):} When the aggregate decay is sufficient, the complexity scales linearly in $k$ as $k \log n$ away from the critical endpoint. In particular, a flat view can be accommodated by a sufficiently concentrated partner, illustrating the compensatory mechanism of the bilateral pursuit.
\end{itemize}

At the critical configuration $\{\alpha_\bfu,\alpha_\bfv\}=\{0,1\}$, the additional $\log k$ factor comes from the harmonic normalization $H_{k,1}\asymp\log k$ at decay exponent one. Thus, the polynomial transition remains at $\alpha_\bfu+\alpha_\bfv=1$, while this endpoint carries a logarithmic correction in the present sufficient bound. This correction can be noticeable at moderate sparsity, although it is asymptotically smaller than any positive polynomial power of $k$. From a signal processing perspective, the transition captures a compensatory effect across views: a concentrated direction in one view provides an effective filter that suppresses ambient cross-covariance noise and facilitates recovery of a more diffuse partner direction.

\subsection{One-Step Thresholded Post-Refinement}\label{sec:post-refinement}
We now return to the approximation term in \Cref{thm:recovery_guarantee}. The direct stagewise analysis provides constant-level alignment but retains $\sqrt{1-\gamma}$. Starting from this alignment, one symmetric thresholded-power update in each view can further refine the directions. We show that these two parallel updates remove the approximation term under the same sample-complexity order.

Let $\mathcal H_k(\bfx)$ retain the $k$ largest-magnitude coordinates of $\bfx$ and set all remaining coordinates to zero, using a fixed deterministic rule to break ties. Let $(\hat{\bfu},\hat{\bfv})$ denote the output of \Cref{alg:bi_sep}. Using the same sample cross-covariance matrix, define
\begin{equation}\label{eq:post-refinement}
	\begin{aligned}
		\tilde{\bfu}
		&=\frac{\mathcal H_{k_\bfu}(\bhSigma_{xy}\hat{\bfv})}
		{\|\mathcal H_{k_\bfu}(\bhSigma_{xy}\hat{\bfv})\|_2},\\
		\tilde{\bfv}
		&=\frac{\mathcal H_{k_\bfv}(\bhSigma_{xy}^{\top}\hat{\bfu})}
		{\|\mathcal H_{k_\bfv}(\bhSigma_{xy}^{\top}\hat{\bfu})\|_2}.
	\end{aligned}
\end{equation}
This post-refinement requires only two matrix--vector multiplications and two top-coordinate selections. Its key mechanism is that the final Bi-SEP estimates already have constant alignment with both true directions, so each update consists of a sparse signal component with amplitude at least $\rho\gamma$ plus a projected noise term controlled by the uniform sparse-block event. Once this noise is sufficiently smaller than the signal component, hard thresholding converts the constant-alignment estimate into a vanishing statistical-error bound. Under the conditions below, both denominators in \eqref{eq:post-refinement} are strictly positive.

\begin{corollary}[One-step post-refinement]\label{cor:post-refinement}
	Under the setting of \Cref{thm:recovery_guarantee}, suppose its sample-complexity condition holds with a sufficiently large absolute numerical constant, and let $(\tilde{\bfu},\tilde{\bfv})$ be defined by \eqref{eq:post-refinement}. Conditioned on $\mcE$, there exists a numerical constant $C_3>0$ such that
	\begin{equation}\label{eq:post-refinement-bound}
		\max\!\left\{\sin\angle(\tilde{\bfu},\bfu),\sin\angle(\tilde{\bfv},\bfv)\right\}
		\le C_3\frac{1+\rho}{\rho\gamma}\sqrt{\frac{(\ku+\kv)\log n}{m}}.
	\end{equation}
\end{corollary}

%% file: figs/complexity_phase_transit.tex
\begin{figure}[t]
	\centering

	\begin{minipage}{0.48\columnwidth}
		\centering
		\begin{tikzpicture}[scale=2.5, >=stealth]

			\def\xmax{1.3}
			\def\ymax{1.3}

			\fill[linearColor] (1,0) -- (\xmax,0) -- (\xmax,\ymax) -- (0,\ymax) -- (0,1) -- cycle;
			\fill[superlinearColor] (0,0) -- (1,0) -- (0,1) -- cycle;

			\draw[very thick, boundaryRed, dashed] (1,0) -- (0,1);

			\node[textDark, align=center, font=\scriptsize] at (0.34, 0.22) {
				\textbf{Super-linear}\\
				(Weak Synergy)\\
				$\alpha_\bfu + \alpha_\bfv < 1$
			};

			\node[textDark, align=center, font=\scriptsize, fill=linearColor, inner sep=2pt] at (0.88, 0.9) {
				\textbf{Linear Exponent}\\
				(Strong Synergy)\\
				$\alpha_\bfu + \alpha_\bfv \ge 1$
			};

			\node[boundaryRed, rotate=-45, font=\scriptsize, anchor=south] at (0.5, 0.5) {Phase Transition $\alpha_\bfu + \alpha_\bfv = 1$};

			\draw[->, thick, textDark] (0,0) -- (\xmax+0.1, 0) node[below, font=\scriptsize] {$\alpha_\bfu$};
			\draw[->, thick, textDark] (0,0) -- (0, \ymax+0.1) node[left, font=\scriptsize] {$\alpha_\bfv$};

			\foreach \x in {0.5, 1}
			\draw[thick, textDark] (\x, 0.02) -- (\x, -0.02) node[below, font=\scriptsize] {\x};
			\foreach \y in {0.5, 1}
			\draw[thick, textDark] (0.02, \y) -- (-0.02, \y) node[left, font=\scriptsize] {\y};
			\node[below left, textDark, font=\scriptsize] at (0,0) {0};

		\end{tikzpicture}
		\\[1ex]
		{\small (a) Bi-SEP (Ours)}
	\end{minipage}\hfill
	\begin{minipage}{0.48\columnwidth}
		\centering
		\begin{tikzpicture}[scale=2.5, >=stealth]

			\def\xmax{1.3}
			\def\ymax{1.3}

			\fill[superlinearColor] (0,0) rectangle (\xmax,\ymax);

			\node[textDark, align=center, font=\scriptsize] at (\xmax/2, \ymax/2) {
				\textbf{Quadratic}\\
				(Profile-Blind)\\
				$m \gtrsim \ku\kv \log n$
			};

			\draw[->, thick, textDark] (0,0) -- (\xmax+0.1, 0) node[below, font=\scriptsize] {$\alpha_\bfu$};
			\draw[->, thick, textDark] (0,0) -- (0, \ymax+0.1) node[left, font=\scriptsize] {$\alpha_\bfv$};

			\foreach \x in {0.5, 1}
			\draw[thick, textDark] (\x, 0.02) -- (\x, -0.02) node[below, font=\scriptsize] {\x};
			\foreach \y in {0.5, 1}
			\draw[thick, textDark] (0.02, \y) -- (-0.02, \y) node[left, font=\scriptsize] {\y};
			\node[below left, textDark, font=\scriptsize] at (0,0) {0};

		\end{tikzpicture}
		\\[1ex]
		{\small (b) Standard Hard-Thresholding}
	\end{minipage}

	\caption{Sample-complexity scaling regimes under power-law decay profiles. \textbf{(a)} Bi-SEP leverages signal concentration, producing a transition at $\alpha_\bfu+\alpha_\bfv=1$; the green region denotes linear scaling in the polynomial order. \textbf{(b)} The displayed profile-blind hard-thresholding scaling remains quadratic throughout the parameter space.}
	\label{fig:complexity_phase_transit}
\end{figure}

%% file: 04Proofs.tex
\section{Proofs of Main Results}
\label{sec:proofs}

In this section, we present the formal proofs of our main theoretical results. We begin by establishing a series of foundational lemmas and propositions that characterize the uniform noise landscape, the signal structure properties, and the iterative dynamics of the Bi-SEP algorithm. These elements serve as the essential building blocks for our main inductive proof. The detailed mathematical proofs of all supporting lemmas and propositions in this subsection are deferred to the appendices.

\subsection{Foundational Propositions and Lemmas}
We begin by formally stating the probabilistic guarantee for the uniform noise event $\mathcal{E}(t)$ introduced in \Cref{sec:results}.
\begin{proposition}[Uniform spectral norm control over sparse blocks]\label{prop:union_event}
	Under the condition that $m \gtrsim (|S_1|+|S_2|)\log n+t$, there exists a constant $c>0$ such that $\mathbb P(\mcE(t)) \ge 1-2e^{-ct}$.
	Particularly, letting $t=\log n$, the uniform noise event $\mcE$ holds with probability at least $1-2n^{-c}$.
\end{proposition}
Conditioned on this high-probability event $\mcE$, the specific sample complexity at each iteration is dictated by the signal's energy concentration. It is worth noting that although we adopt the Gaussian model in the main statement for consistency with the standard setup in the literature, the concentration property over sparse blocks in this proposition is not specific to Gaussian data. Analogous bounds hold for broader classes of distributions, including suitable sub-Gaussian designs. Thus, whenever this proposition holds, possibly with modified constants and tail parameters, the subsequent Bi-SEP analysis remains valid.

To streamline the subsequent inductive analysis, we then introduce a lemma which shows that shifting indices by a single algorithmic step (e.g., $t$ vs. $t+1$) affects the sample complexity bounds only by an absolute constant factor. This allows us to absorb specific transition constants into a global constant $C$.
\begin{lemma}[Asymptotically equivalent conditions, {\cite[Lemma 1]{xu2025spca}}]\label{lem:asymp-equiv}
	Let $s(\cdot)$ be a signal energy structure function as per Definition \ref{def:structure_function}. The following conditions are equivalent in the asymptotic sense:
	\begin{enumerate}[label=(\arabic*)]
		\item \(m\ge C_1\, ps(p)\);
		\item \(m \ge C_2\, (p+1)s(p)\);
		\item \(m \ge C_3\, ps(p+1)\).
	\end{enumerate}
	In other words, if one of these conditions holds for some absolute constant \(C_i>0\), then the other two also hold for some (different) absolute constants \(C_j>0\).
\end{lemma}

Next, we establish the three core components that govern the algorithmic execution of Bi-SEP. To initialize the algorithmic induction, we first quantify the energy securely captured during the zero-step refinement.

\begin{proposition}[Initialization Guarantee]
	\label{prop:initialization}
	Condition on the high-probability event $\mcE$. Let $(i_0, j_0) = \arg\max_{i,j} |(\hat{\Sigma}_{xy})_{ij}|$ be the indices selected in the initialization step. For any $\gamma \in (0,1)$, if the sample size satisfies
	\begin{equation}\label{eq:initial-t-0-sample}
		m \ge \frac{C (1+\rho)^2}{\rho^2 (1-\sqrt{\gamma})^2} s_\bfu(1) s_\bfv(1) \log n,
	\end{equation}
	then the selected coordinates capture sufficient signal energy product:
	\begin{equation}
		|u_{i_0}| |v_{j_0}| \ge \sqrt{\frac{\gamma}{s_\bfu(1) s_\bfv(1)}}.
	\end{equation}
\end{proposition}

Following initialization, the iterative expansion relies on a bilateral coupling mechanism. The first component dictates how well a restricted estimator aligns with the true canonical vector in the presence of bounded cross-covariance noise.

\begin{proposition}[Alignment-Energy Coupling]
	\label{prop:alignment}
	Let $\hat{\bfu}$ and $\hat{\bfv}$ be the unit leading left and right singular vectors of the restricted empirical cross-covariance matrix $(\bhSigma_{xy})_{S_\bfu, S_\bfv}$, respectively, padded with zeros outside their respective supports. Specifically, $\|\hat{\bfu}\|_2 = \|\hat{\bfv}\|_2 = 1$, and $\text{supp}(\hat{\bfu}) \subseteq S_\bfu$, $\text{supp}(\hat{\bfv}) \subseteq S_\bfv$. If the spectral gap condition $\rho \|\bfu_{S_\bfu}\|_2 \|\bfv_{S_\bfv}\|_2 > 2\|\bfW_{S_\bfu, S_\bfv}\|_2$ holds, then the alignment of the estimators with the true vectors satisfies
	\begin{equation}\label{eq:u-align}
		|\langle \hat{\bfu}, \bfu \rangle| \ge \|\bfu_{S_\bfu}\|_2 \sqrt{1 - \left(\frac{2\|\bfW_{S_\bfu, S_\bfv}\|_2}{\rho \|\bfu_{S_\bfu}\|_2 \|\bfv_{S_\bfv}\|_2}\right)^2}
	\end{equation}
	and
	\begin{equation}\label{eq:v-align}
		|\langle \hat{\bfv}, \bfv \rangle| \ge \|\bfv_{S_\bfv}\|_2 \sqrt{1 - \left(\frac{2\|\bfW_{S_\bfu, S_\bfv}\|_2}{\rho \|\bfu_{S_\bfu}\|_2 \|\bfv_{S_\bfv}\|_2}\right)^2}.
	\end{equation}
\end{proposition}

The second component of the iterative cycle ensures that an estimate with sufficient alignment in one view enables the proxy vector to accurately capture the unrecovered signal energy in the partner view.

\begin{proposition}[Cross-Coordinate Energy Boosting]
	\label{prop:reselection}
	Consider the update of vector $\hat{\bfu}^{(t)} \to \hat{\bfu}^{(t+1)}$ at step $t$. Let $\hat{\bfv}^{(t)}$ be the current partner estimator supported on $S_\bfv^{(t)}$ with alignment $|\langle \bfv, \hat{\bfv}^{(t)} \rangle|$. Recall that $S_\bfu^{(t+1)}$ is the support selected by the proxy $\bfr_\bfu = \bhSigma_{xy} \hat{\bfv}^{(t)}$. Condition on event $\mcE$. Then we have
	\begin{align}
		\nonumber & \|\bfu_{S_\bfu^{(t+1)}}\|_2                                                                                                                                                                  \\
		\ge ~     & \sqrt{\frac{1}{s_\bfu(|S_\bfu^{(t+1)}|)}} - \frac{2C(1+\rho)}{\rho |\langle \bfv, \hat{\bfv}^{(t)} \rangle|} \sqrt{\frac{(|S_\bfu^{(t+1)}| + |S_\bfv^{(t)}|) \log n}{m}}.\label{eq:u-vbound}
	\end{align}
	Similarly, for the symmetric update $\hat{\bfv}^{(t)} \to \hat{\bfv}^{(t+1)}$, we have
	\begin{align}
		\nonumber & \|\bfv_{S_\bfv^{(t+1)}}\|_2                                                                                                                                                                  \\
		\ge ~     & \sqrt{\frac{1}{s_\bfv(|S_\bfv^{(t+1)}|)}} - \frac{2C(1+\rho)}{\rho |\langle \bfu, \hat{\bfu}^{(t)} \rangle|} \sqrt{\frac{(|S_\bfu^{(t)}| + |S_\bfv^{(t+1)}|) \log n}{m}}.\label{eq:v-ubound}
	\end{align}
\end{proposition}

Finally, to analytically instantiate our general sample complexity bounds for specific signal structures (as required for the proof of the phase transition corollary), we formalize the explicit asymptotic growth rates of power-law structure functions.

\begin{lemma}[Structure function of power-law signals, {\cite[Proposition 2]{xu2025spca}}]
	\label{lem:power_law_sp}
	Suppose a signal with sparsity $k$ follows a power-law decay profile with rate $\alpha$ (Definition \ref{def:power_law}). Its structure function $s(p)$ for $p\le k$ satisfies the following asymptotic growth:
	\begin{equation}
		s(p) \asymp
		\begin{cases}
			\left(\frac{k}{p}\right)^{1-\alpha}, & 0 \le \alpha < 1, \\[6pt]
			\frac{1+\log k}{1+\log p},           & \alpha = 1,       \\[6pt]
			1,                                   & \alpha > 1.
		\end{cases}
	\end{equation}
\end{lemma}

With the foundational properties established, we proceed to the formal proofs of the main theoretical results. The proof of \Cref{thm:recovery_guarantee} is structured inductively: it relies on \Cref{prop:initialization} to establish a statistically reliable base case, and synthesizes the bilateral coupling mechanisms (\Cref{prop:alignment} and \Cref{prop:reselection}) to drive the iterative steps. We then apply the asymptotic bounds from \Cref{lem:power_law_sp} to derive \Cref{cor:scca_power_law_transition}, before proving the one-step post-refinement guarantee in \Cref{cor:post-refinement}.

\subsection{Proof of Theorem \ref{thm:recovery_guarantee}}
Throughout the proof, absolute numerical constants may change from line to line and are absorbed into the global constant $C$ appearing in the sample complexity condition. The proof proceeds by mathematical induction on the iteration step $t$. We aim to show that for all $1 \le t \le k_{\max}$, the estimators capture sufficient energy
\begin{equation}\label{eq:energy_bound}
	\|\bfu_{S_\bfu^{(t)}}\|_2 \ge \sqrt{\frac{\gamma}{s_\bfu(t_\bfu)}} \quad \text{and} \quad \|\bfv_{S_\bfv^{(t)}}\|_2 \ge \sqrt{\frac{\gamma}{s_\bfv(t_\bfv)}},
\end{equation}
where $t_\bfu = \min(t, k_\bfu)$ and $t_\bfv = \min(t, k_\bfv)$ denote the effective support sizes at step $t$.

\textbf{1. Base Case ($t=0 \to 1$):}
By Proposition \ref{prop:initialization}, the initialization selects indices $(i_0, j_0)$ satisfying the joint product bound
\begin{equation}\label{eq:t-0-bound}
	|u_{i_0}v_{j_0}| \ge \sqrt{\frac{\gamma}{s_\bfu(1)s_\bfv(1)}}.
\end{equation}

Consider the first update for $\bfu$. Starting with the initial proxy support $|S_\bfu^{(1)}| = 1$, Proposition \ref{prop:reselection} dictates that the captured energy satisfies
\begin{equation}
	\|\bfu_{S_\bfu^{(1)}}\|_2 \ge \frac{1}{\sqrt{s_\bfu(1)}} - \frac{2C(1+\rho)}{\rho |v_{j_0}|} \sqrt{\frac{2\log n}{m}}.
\end{equation}
To guarantee the target energy bound $\|\bfu_{S_\bfu^{(1)}}\|_2 \ge \sqrt{\gamma/s_\bfu(1)}$, the noise term must be strictly bounded by $(1-\sqrt{\gamma})|u_{(1)}|$, where $|u_{(1)}| = 1/\sqrt{s_\bfu(1)}$ denotes the maximum signal magnitude. Rearranging this condition for $m$ yields
\begin{equation}\label{eq:m-init-intermediate}
	m \ge \frac{8 C^2 (1+\rho)^2}{\rho^2 (1-\sqrt{\gamma})^2} \frac{\log n}{|v_{j_0}|^2 |u_{(1)}|^2}.
\end{equation}

Here lies the crucial decoupling step. By definition of the maximum magnitude, $|u_{(1)}| \ge |u_{i_0}|$ holds deterministically. Coupling this property with the initialization guarantee \eqref{eq:t-0-bound}, we establish a lower bound for the asymmetric cross-view signal product:
\begin{equation}
	|u_{(1)}| |v_{j_0}| \ge |u_{i_0}| |v_{j_0}| \ge \sqrt{\frac{\gamma}{s_\bfu(1)s_\bfv(1)}}.
\end{equation}

Since this product appears in the denominator of \eqref{eq:m-init-intermediate}, substituting its lower bound yields a stricter but sufficient condition for the sample size. The requirement simplifies to
\begin{equation}
	m \ge \frac{8 C^2 (1+\rho)^2}{\rho^2 \gamma (1-\sqrt{\gamma})^2} s_\bfu(1)s_\bfv(1) \log n.
\end{equation}

This condition is strictly subsumed by the global sample complexity assumption \eqref{eq:m-bound-precise}. A symmetric argument applies to the first update of $\bfv$. Thus, the base case \eqref{eq:energy_bound} is rigorously established for $t=1$.

\textbf{2. Inductive Step ($t \to t+1$):}
Assume the energy bounds~\eqref{eq:energy_bound} hold at step $t$. We now show they persist for step $t+1$. We focus on the update for $\bfu$; the argument for $\bfv$ is identical.

\textit{Step 2a: Alignment of the partner estimator.}
First, we establish the alignment of $\hat{\bfv}^{(t)}$. Under the inductive hypothesis, the signal strength on the current supports is lower bounded by $\sigma^*_t = \rho \|\bfu_{S_\bfu^{(t)}}\|_2 \|\bfv_{S_\bfv^{(t)}}\|_2 \ge \rho \gamma / \sqrt{s_\bfu(t_\bfu)s_\bfv(t_\bfv)}$.

To enable the boosting step, we require $|\langle \hat{\bfv}^{(t)}, \bfv \rangle| \ge \sqrt{\gamma} \|\bfv_{S_\bfv^{(t)}}\|_2$. By Proposition \ref{prop:alignment}, it suffices to ensure that the noise term satisfies $2\|\bfW_{S_\bfu^{(t)}, S_\bfv^{(t)}}\|_2 / \sigma^*_t \le \sqrt{1-\gamma}$. Applying the uniform bound~\eqref{eq:union_event} for $\|\bfW_{S_\bfu^{(t)}, S_\bfv^{(t)}}\|_2$ under event $\mcE$ and substituting the lower bound for $\sigma^*_t$, this condition translates to
\begin{equation}
	m \ge \frac{4 C^2 (1+\rho)^2}{\rho^2 \gamma^2 (1-\gamma)} (t_\bfu+t_\bfv) s_\bfu(t_\bfu)s_\bfv(t_\bfv) \log n.
\end{equation}

Since $(1-\sqrt{\gamma})^2 < 1-\gamma$ for $\gamma \in (0,1)$, this requirement is strictly subsumed by the global sample complexity condition \eqref{eq:m-bound-precise}. Consequently, the alignment bound holds, yielding
\begin{equation}\label{eq:alignment_lower_bound}
	|\langle \hat{\bfv}^{(t)}, \bfv \rangle| \ge \sqrt{\gamma} \|\bfv_{S_\bfv^{(t)}}\|_2 \ge \frac{\gamma}{\sqrt{s_\bfv(t_\bfv)}}.
\end{equation}

\textit{Step 2b: Reselection of the next support.}
Now consider the update $\hat{\bfv}^{(t)} \to \hat{\bfu}^{(t+1)}$. For notational simplicity during this step, let $p = |S_\bfu^{(t+1)}|$ and $q = |S_\bfv^{(t)}|$ denote the respective support sizes. Applying Proposition \ref{prop:reselection}, we have
\begin{equation}
	\|\bfu_{S_\bfu^{(t+1)}}\|_2
	\ge \frac{1}{\sqrt{s_\bfu(p)}} - \frac{2C(1+\rho)}{\rho |\langle \hat{\bfv}^{(t)}, \bfv \rangle|} \sqrt{\frac{(p + q) \log n}{m}}.
\end{equation}
To guarantee $\|\bfu_{S_\bfu^{(t+1)}}\|_2 \ge \sqrt{\gamma/s_\bfu(p)}$, it suffices to bound the subtraction term by $(1-\sqrt{\gamma})/\sqrt{s_\bfu(p)}$. Substituting the explicit alignment lower bound \eqref{eq:alignment_lower_bound} for $|\langle \hat{\bfv}^{(t)}, \bfv \rangle|$ and rearranging for $m$ yields the requirement
\begin{equation}
	m \ge \frac{4C^2(1+\rho)^2}{\rho^2 \gamma^2 (1-\sqrt{\gamma})^2} \left[ (p + q) s_\bfu(p) s_\bfv(q) \right] \log n.
\end{equation}
This condition matches the general term in the sample complexity condition \eqref{eq:m-bound-precise} (utilizing the equivalence established in Lemma \ref{lem:asymp-equiv} for index shifts). Since the theorem assumes this holds for the maximum over all $t$, the energy bound for $\bfu$ at $t+1$ is established.

\textbf{3. Final Estimation Error:}
At the final iteration $k_{\max}$, the estimators are formed using the supports $S_\bfu^{(k_{\max})}$ and $S_\bfv^{(k_{\max})}$. By the inductive proof, these supports capture a significant fraction of the signal energy: $\|\bfu_{S_\bfu^{(k_{\max})}}\|_2 \ge \sqrt{\gamma}$ and $\|\bfv_{S_\bfv^{(k_{\max})}}\|_2 \ge \sqrt{\gamma}$.

To bound the total estimation error, we decompose the principal angle using the triangle inequality. Let $\bar{\bfu} = \bfu_{S_\bfu^{(k_{\max})}} / \|\bfu_{S_\bfu^{(k_{\max})}}\|_2$ denote the normalized projection of $\bfu$ onto the final support. Then, we have:
\begin{equation}
	\sin\angle(\hat{\bfu}^{(k_{\max})}, \bfu) \le \sin\angle(\bar{\bfu}, \bfu) + \sin\angle(\hat{\bfu}^{(k_{\max})}, \bar{\bfu}).
\end{equation}

The first term represents the approximation error due to unselected coordinates, which satisfies $\sin\angle(\bar{\bfu}, \bfu) = \sqrt{1 - \|\bfu_{S_\bfu^{(k_{\max})}}\|_2^2} \le \sqrt{1-\gamma}$. The second term is the statistical estimation error within the captured subspace. Applying \Cref{cor:wedin_rank1_weyl} (detailed in Appendix \ref{sec:app-A}) with the restricted signal strength $\sigma_1 = \rho \|\bfu_{S_\bfu^{(k_{\max})}}\|_2 \|\bfv_{S_\bfv^{(k_{\max})}}\|_2 \ge \rho\gamma$, we obtain
\begin{align}
	\nonumber\sin\angle(\hat{\bfu}^{(k_{\max})}, \bar{\bfu}) & \le \frac{2\|\bfW_{S_\bfu^{(k_{\max})},S_\bfv^{(k_{\max})}}\|_2}{\sigma_1}  \\
	                                                         & \le \frac{2 C(1+\rho)}{\rho \gamma} \sqrt{\frac{(k_\bfu+k_\bfv)\log n}{m}}.
\end{align}

Combining these two bounds yields the final error for $\bfu$ in \eqref{eq:u-error}. A symmetric argument establishes the bound for $\bfv$ in \eqref{eq:v-error}, completing the proof. \hfill \IEEEQEDopen

\subsection{Proof of Corollary \ref{cor:scca_power_law_transition}}

As established in \eqref{eq:m-bound-asymp}, the asymptotic sample complexity in the balanced regime ($k_\bfu = k_\bfv = k$) is dictated by the maximum of the complexity term $\mathcal{C}(p) := p s_\bfu(p) s_\bfv(p)$ over $1 \le p \le k$.

First, consider the regime where $\alpha_\bfu \neq 1$ and $\alpha_\bfv \neq 1$.
By Lemma \ref{lem:power_law_sp}, the structure functions scale as $s(p) \asymp (k/p)^{1-\bar{\alpha}}$, where the saturated decay rate is defined as $\bar{\alpha} = \min(\alpha, 1)$. Substituting these asymptotics yields
\begin{equation}
	\mathcal{C}(p) \asymp p^{\bar{\alpha}_{\bfu\bfv}-1} k^{2-\bar{\alpha}_{\bfu\bfv}}, \quad \text{where } \bar{\alpha}_{\bfu\bfv} := \bar{\alpha}_\bfu + \bar{\alpha}_\bfv.
\end{equation}
The maximum is determined by the behavior of the power function $p^{\bar{\alpha}_{\bfu\bfv}-1}$:
\begin{itemize}
	\item If $\bar{\alpha}_{\bfu\bfv} < 1$, the function is strictly decreasing, peaking at $p=1$ with order $k^{2-\bar{\alpha}_{\bfu\bfv}}$. Since $\alpha_\bfu, \alpha_\bfv \ge 0$, the condition $\bar{\alpha}_{\bfu\bfv} < 1$ holds if and only if $\alpha_\bfu + \alpha_\bfv < 1$. In this regime, no truncation occurs (i.e., $\bar{\alpha}_{\bfu\bfv} = \alpha_\bfu + \alpha_\bfv$), yielding the complexity order $k^{2-(\alpha_\bfu + \alpha_\bfv)}$.
	\item If $\bar{\alpha}_{\bfu\bfv} \ge 1$, the function is non-decreasing, peaking at $p=k$ with order $k$. Mathematically, this condition is exactly equivalent to $\alpha_\bfu + \alpha_\bfv \ge 1$.
\end{itemize}
Combining these two disjoint cases, the final exponent seamlessly simplifies to the untruncated form $\tau = \max(1, 2 - (\alpha_\bfu + \alpha_\bfv))$.

Next, consider the boundary case where $\alpha_\bfu = 1$ (the symmetric case $\alpha_\bfv = 1$ follows identical logic).
Here, the structure function scales as $s_\bfu(p) \asymp \frac{1+\log k}{1+\log p}$. The complexity term becomes $\mathcal{C}(p) \asymp p \frac{1+\log k}{1+\log p} s_\bfv(p)$.
\begin{itemize}
	\item If $\alpha_\bfv = 0$, then $s_\bfv(p) \asymp k/p$. Thus, $\mathcal{C}(p) \asymp k \frac{1+\log k}{1+\log p}$, which is strictly decreasing with respect to $p$. The maximum is attained at $p=1$, yielding the order $k \log k$.
	\item If $\alpha_\bfv > 0$, then $s_\bfv(p) \lesssim (k/p)^{1-\bar{\alpha}_\bfv}$ for some $\bar{\alpha}_\bfv \in (0, 1]$. In this case, $\mathcal{C}(p) \lesssim k^{1-\bar{\alpha}_\bfv} \frac{p^{\bar{\alpha}_\bfv}}{1+\log p} (1+\log k)$. Since $\bar{\alpha}_\bfv > 0$, the algebraic growth of $p^{\bar{\alpha}_\bfv}$ asymptotically dominates the logarithmic decay of $(1+\log p)^{-1}$. Consequently, for sufficiently large $k$, the supremum over $p \in [1, k]$ is attained at the endpoint $p=k$, yielding order $k$.
\end{itemize}
This completes the derivation of the phase transition exponent, explicitly exposing the additional $\log k$ penalty that arises exclusively in the extreme boundary configuration $\{\alpha_\bfu, \alpha_\bfv\} = \{0, 1\}$.
\hfill \IEEEQEDopen

\subsection{Proof of \Cref{cor:post-refinement}}\label{sec:proof-post-refinement}
\begin{proof}
	Write $(\hat\bfu,\hat\bfv)=(\hat\bfu^{(k_{\max})},\hat\bfv^{(k_{\max})})$. At the final stage, the energy induction gives $\|\bfu_{S_\bfu^{(k_{\max})}}\|_2\ge\sqrt\gamma$ and $\|\bfv_{S_\bfv^{(k_{\max})}}\|_2\ge\sqrt\gamma$. The argument leading to \eqref{eq:alignment_lower_bound}, together with its symmetric counterpart and evaluated at the terminal support sizes, therefore yields
	\begin{equation}\label{eq:post-initial-alignment}
		\min\bigl\{|\langle\hat\bfu,\bfu\rangle|,|\langle\hat\bfv,\bfv\rangle|\bigr\}\ge\gamma.
	\end{equation}
	Let
	\begin{equation}
		b=C(1+\rho)\sqrt{\frac{(\ku+\kv)\log n}{m}}.
	\end{equation}
	For any $A\subset[n]$ with $|A|\le\ku$, the sparsity of $\hat\bfv$ and the event $\mcE$ imply
	\begin{equation}
		\|(\bfW\hat\bfv)_A\|_2
		\le \|\bfW_{A,\operatorname{supp}(\hat\bfv)}\|_2\le b.
	\end{equation}
	Let $\sigma_\bfv=\operatorname{sign}(\langle\bfv,\hat\bfv\rangle)$. Since multiplying a vector by $\sigma_\bfv$ does not change its thresholded direction up to sign, we may apply \Cref{lem:threshold-refinement} to
	\begin{equation}
		\sigma_\bfv\bhSigma_{xy}\hat\bfv
		=\rho|\langle\bfv,\hat\bfv\rangle|\bfu+\sigma_\bfv\bfW\hat\bfv,
	\end{equation}
	with signal amplitude $a\ge\rho\gamma$. At $t=k_{\max}$ we have $s_\bfu(k_\bfu)=s_\bfv(k_\bfv)=1$, so the terminal term in \eqref{eq:m-bound-precise}, with a sufficiently large absolute numerical constant, ensures $b\le\rho\gamma/6$. Therefore,
	\begin{equation}\label{eq:post-u-bound}
		\sin\angle(\tilde\bfu,\bfu)
		\le C\frac{1+\rho}{\rho\gamma}\sqrt{\frac{(\ku+\kv)\log n}{m}}.
	\end{equation}

	The update for $\tilde\bfv$ is completely symmetric. Indeed, $\hat\bfu$ is $\ku$-sparse, so for every $B\subset[n]$ with $|B|\le\kv$,
	\begin{equation}
		\|(\bfW^\top\hat\bfu)_B\|_2
		\le \|\bfW_{\operatorname{supp}(\hat\bfu),B}\|_2\le b.
	\end{equation}
	Together with $|\langle\hat\bfu,\bfu\rangle|\ge\gamma$, the same application of \Cref{lem:threshold-refinement} to $\bhSigma_{xy}^\top\hat\bfu$ yields the identical bound for $\sin\angle(\tilde\bfv,\bfv)$. Combining the two estimates proves \eqref{eq:post-refinement-bound}.
\end{proof}

%% file: 05Discussion.tex
\section{Discussion}
\label{sec:discussion}
In this section, we first clarify the relationship between Bi-SEP and the SEP framework for SPCA. Then we explain how the guarantee on the joint product obtained at initialization can be converted into an individual energy guarantee that initiates the inductive recovery. Finally, we discuss the origin of the sparsity coupling appearing in the theoretical error bound.

\subsection{Relationship to the SEP Framework for SPCA}\label{sec:discussion-sep}

	Although Bi-SEP is motivated by SEP, there are significant differences in how Bi-SEP handles the cross-view dependencies. We list the components inherited from SEP and the developments required for SCCA as follows.
	\begin{itemize}
		\item \textbf{Components inherited from SEP:}
		      Bi-SEP adopts the signal energy structure function as a continuous measure of cumulative energy concentration and retains the principle of progressively increasing the size of the estimated support while globally reselecting its entries at each iteration.

		\item \textbf{Algorithmic developments for SCCA:}
		      The rectangular cross-covariance structure removes the marginal anchor available from the diagonal of a covariance matrix in SPCA. A maximal empirical cross-covariance entry provides only a guarantee on the joint product, whereas updating either canonical vector requires sufficient alignment in the other. Bi-SEP resolves this bilateral dependency through zero-step proxy decoupling followed by a bilateral pursuit over two interacting supports. For unequal sparsity budgets, it further incorporates staggered saturation-and-hold: once one support reaches its prescribed cardinality, that cardinality is held fixed while its entries continue to be globally reselected and its associated direction is updated as the partner support grows. Consequently, the direction associated with the sparser view can continue to be refined without increasing its support cardinality.

		\item \textbf{Theoretical developments for SCCA:}
		      The analysis must propagate energy and alignment guarantees simultaneously across the two views, including the stages after one support has saturated while the other continues to expand. Accordingly, the two adapted structure functions enter through the coupled product $s_\bfu(p)s_\bfv(q)$ rather than through a single profile. This bilateral induction yields a sample complexity governed by joint energy concentration (\Cref{thm:recovery_guarantee}) and a phase transition governed by aggregate decay (\Cref{cor:scca_power_law_transition}). It thereby formalizes how concentration in one view can compensate for diffuseness in the other.
	\end{itemize}

\subsection{From Coupled Product to Individual Energy}\label{sec:discussion-product}

A central challenge in analyzing SCCA is bridging initialization and the iterative procedure. The iterative analysis relies on a cyclic dependency: alignment requires sufficient signal energy (\Cref{prop:alignment}), while capturing energy requires sufficient alignment in the partner view to suppress noise (\Cref{prop:reselection}). Entering this inductive cycle therefore requires an initial lower bound on $\|\bfu_{S^{(1)}}\|_2$.

However, initialization (\Cref{prop:initialization}) guarantees only a large joint product $|u_{i_0} v_{j_0}|$, rather than a strong individual alignment $|\langle \hat{\bfv}^{(0)}, \bfv \rangle| = |v_{j_0}|$. If the initial probe $v_{j_0}$ is small, standard bounds become insufficient because the noise term may dominate. The proxy generation step at $t=0$ resolves this issue through two mechanisms:

\begin{enumerate}

	\item \textbf{Scalar Multiplication Preserves Ordering:}
	      Since $\hat{\bfv}^{(0)}$ is $1$-sparse (supported on $j_0$), the signal component of the proxy vector $\bfr_\bfu = \bhSigma_{xy}\hat{\bfv}^{(0)}$ is $\rho v_{j_0}\bfu$. The scalar factor $v_{j_0}$ therefore does not affect the ordering of coordinates. Selecting the maximal proxy entry consequently targets the largest coordinate of the true signal, $u_{(1)}$, capturing an effective signal strength $\rho |u_{(1)} v_{j_0}|$.

	\item \textbf{Joint Product Guarantees SNR:}
	      The initialization step selects $(i_0,j_0)$ as the maximal empirical entry, ensuring that the signal term $\rho |u_{i_0} v_{j_0}|$ dominates the noise (\Cref{prop:initialization}). Since $|u_{(1)}| \ge |u_{i_0}|$, it follows that
	      \begin{equation}
		      \rho |u_{(1)} v_{j_0}| \ge \rho |u_{i_0} v_{j_0}|,
	      \end{equation}
	      which guarantees that the true peak remains detectable in the proxy and establishes a valid lower bound for $\|\bfu_{S^{(1)}}\|_2$.

\end{enumerate}

Therefore, even if the initial alignment $|v_{j_0}|$ is weak, the proxy selection remains reliable because the signal-to-noise ratio is determined by the joint product rather than the isolated probe. This step converts the joint initialization guarantee into an individual energy bound, enabling the subsequent inductive analysis. This mechanism contrasts with symmetric settings such as SPCA, where initialization directly isolates coordinates with high energy.

\subsection{On the Coupling in the Error Bound}\label{sec:discussion-coupling}
In \Cref{thm:recovery_guarantee}, the statistical error rate for estimating $\bfu$ scales as $\sqrt{(\ku+\kv)\log n/m}$. In the current analysis, this coupled dependence results from the joint block-SVD step and its uniform perturbation control. Specifically, we bound the spectral norm of the noise matrix $\bfW$ over all sparse rectangles: $\sup_{|S_1|\le \ku, |S_2|\le \kv} \|\bfW_{S_1,S_2}\|_2$. The $\epsilon$-net argument over these rectangles introduces the combinatorial complexity of both supports, yielding the $(\ku+\kv)$ dependence.

This contrasts with prior works that achieve decoupled rates of $\sqrt{\ku\log n/m}$ by employing a one-sided prediction loss, effectively treating $\bfv$ as a nuisance parameter. As noted in~\cite{gao2017sparse}, applying Wedin's $\sin\Theta$ theorem to cross-covariance perturbations leads to spectral-norm control of the noise block, which consequently incurs the $\kv$ dependence and becomes suboptimal for isolated estimation of $\bfu$. This coupling arises because the intermediate estimator $\hat{\bfv}$ is data-dependent and correlated with the empirical noise, preventing direct control of the projected noise term $\bfW\hat{\bfv}$.

%% file: 06Experiments.tex
\section{Experiments}
\label{sec:experiments}

\begin{figure*}[t]
	\centering
	\subfloat[Flat profile\label{fig:ablation-flat}]{\includegraphics[width=0.32\linewidth]{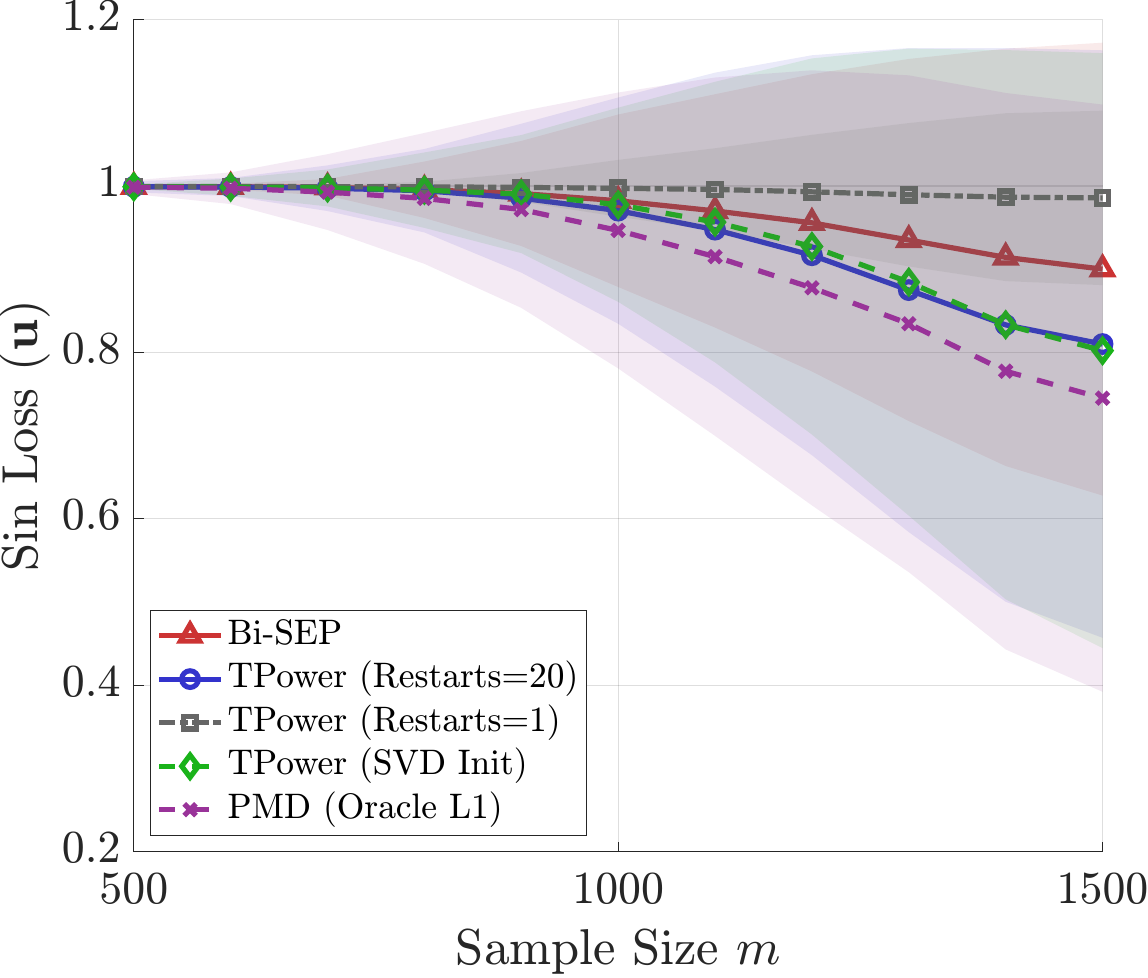}}
	\hfill
	\subfloat[Power-law profile (\(\alpha=1.0\))\label{fig:ablation-power}]{\includegraphics[width=0.32\linewidth]{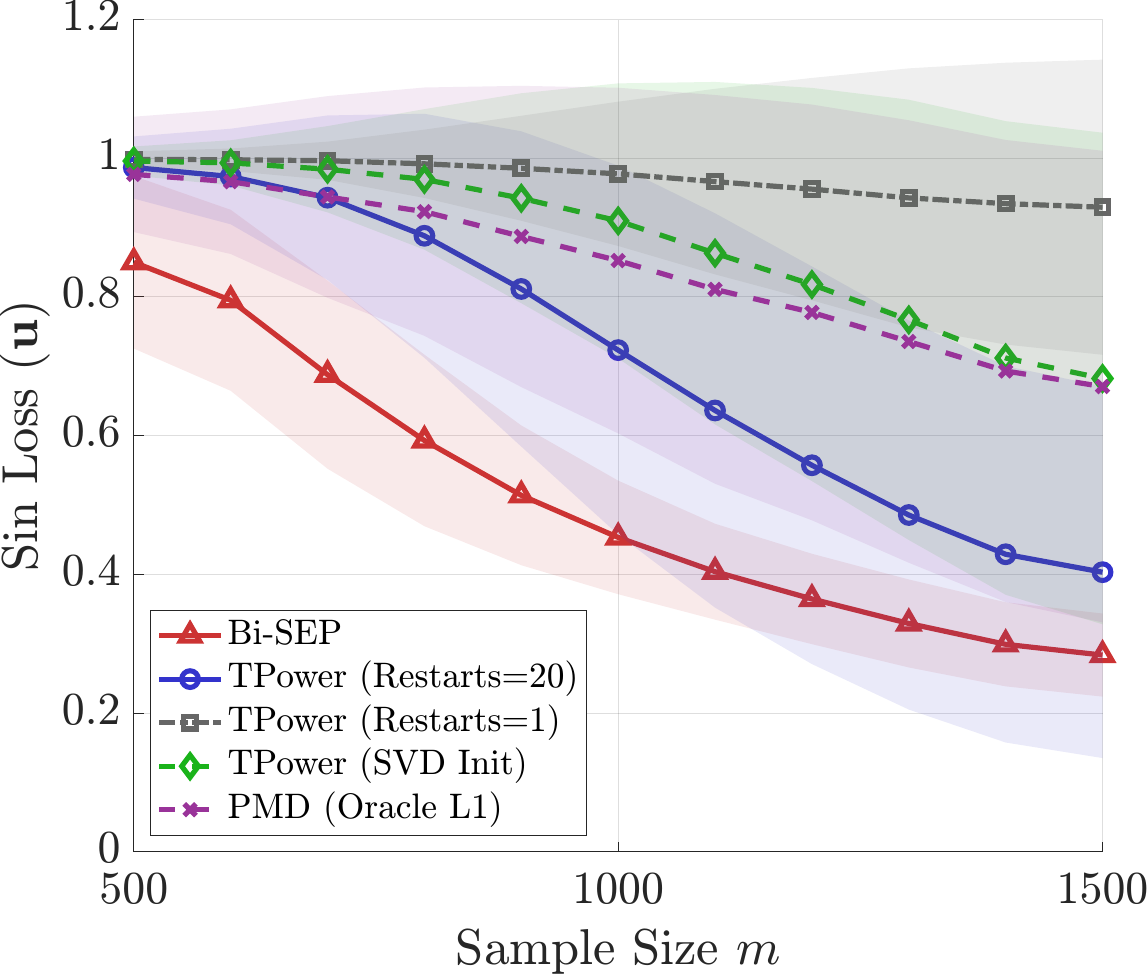}}
	\hfill
	\subfloat[Exponential profile\label{fig:ablation-exp}]{\includegraphics[width=0.32\linewidth]{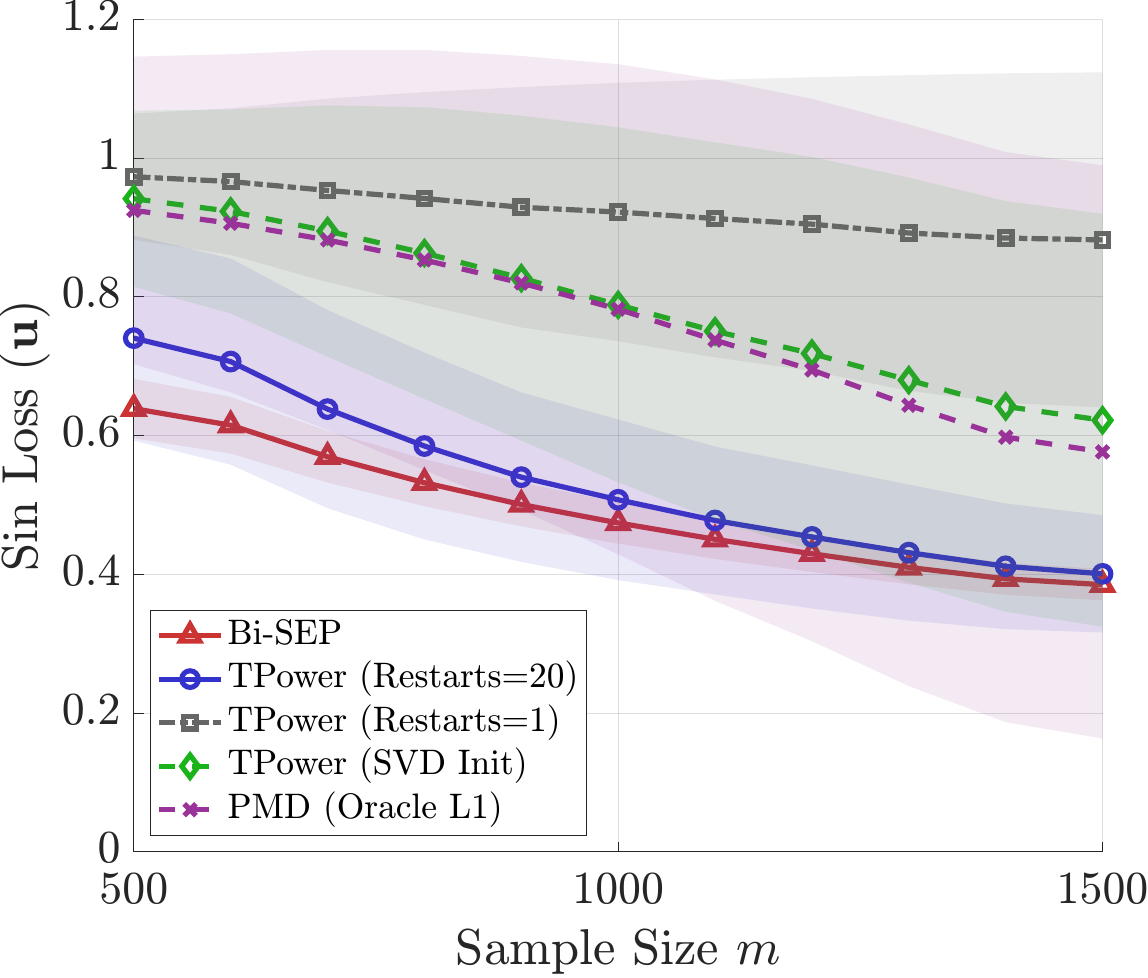}}
	\caption{Estimation error versus sample size under symmetric signal profiles. PMD with oracle tuning and the TPower variants remain competitive for the flat profile. Under nonuniform profiles, Bi-SEP generally has lower error over the displayed sample range.}
	\label{fig:exp1-ablation}
\end{figure*}

	This section evaluates Bi-SEP under the signal profiles and operating regimes considered in the theoretical analysis. The first three experiments examine the profile-dependent behavior, the two-view phase transition, and error behavior in an unbalanced setting. The remaining experiments study sensitivity to the sparsity budgets and canonical correlation, non-Gaussian noise, and a spatial and spectral source localization simulation. All experiments in this section are synthetic. The reported Bi-SEP results use the direct output of \Cref{alg:bi_sep}, without the optional post-refinement in \Cref{sec:post-refinement}. Unless otherwise stated, all experiments use \(n=1000\), \(\rho=0.8\), and \(1000\) Monte Carlo repetitions. We report the sine loss defined in \eqref{eq:sine_loss}; in one-dimensional plots, solid curves show empirical means and shaded regions show \(\pm 1\) empirical standard deviation.

	We compare Bi-SEP with several sparse CCA baselines. TPower is evaluated with one random initialization (TPower-1), with \(20\) random restarts (TPower-20), and with initialization by the leading singular vectors followed by hard thresholding (TPower-SVD). We also include Penalized Matrix Decomposition (PMD)~\cite{witten2009penalized}. Bi-SEP uses cardinality budgets, whereas PMD uses \(\ell_1\) constraints. We set these constraints to the \(\ell_1\) norms of the true canonical vectors, giving PMD oracle tuning.

	\subsection{Profile-Adaptive Sample Complexity}
	\label{sec:exp-sample-complexity}

	We first examine the relation between the sample size \(m\) and the energy profile of the canonical vectors. We set \(k_\bfu=k_\bfv=20\) and consider three profiles used symmetrically in the two views: a flat profile (\(\alpha=0\)), a power-law profile with \(\alpha=1.0\), and an exponentially decaying profile satisfying \(|x_{(i)}|^2\propto \exp(-i)\) before normalization. Because the two views use the same profile and sparsity, \Cref{fig:exp1-ablation} reports the error only for \(\bfu\).

	The flat profile in \Cref{fig:ablation-flat} serves as a difficult reference case, in which PMD with oracle tuning and TPower-20 outperform Bi-SEP. This behavior is consistent with the theory: for flat signals, the sufficient sample size bound reduces to the quadratic order in the sparsity. At a fixed sample size, the observed errors also depend on constants, lower order effects, and optimization behavior. TPower-20 uses multiple random restarts, and PMD uses oracle tuning; both choices can improve their results in this regime.

	For concentrated signals, the ordering changes. In the power law case of \Cref{fig:ablation-power}, Bi-SEP has a lower error over most of the displayed sample range. At \(m=1500\), its mean error is approximately \(0.27\), compared with approximately \(0.38\) for TPower-20 and \(0.65\) for both TPower-SVD and PMD. In the exponential case of \Cref{fig:ablation-exp}, Bi-SEP and TPower-20 are closer at larger sample sizes, while Bi-SEP remains competitive with and generally improves over TPower-1, TPower-SVD, and PMD. These results associate the benefit of Bi-SEP with nonuniform energy profiles.

	The relatively weak performance of TPower-SVD reflects the high-dimensional inconsistency of unregularized spectral initialization. The empirical cross-covariance \(\bhSigma_{xy}\) contains a rank-one signal perturbed by dense estimation noise. When the sample size is smaller than or comparable to the ambient dimension, its leading empirical singular vectors need not consistently estimate the population directions; in the strongly undersampled regime \(m\ll n\), their asymptotic overlap can vanish~\cite{Johnstone2009SPCAConsistency,Paul2007SPCAeigenstructure,BenaychGeorges2012Rectangular}. Consequently, applying hard thresholding to singular vectors contaminated by this noise need not recover the true sparse supports, because their largest empirical coordinates may be generated by dense noise rather than the signal. Thus, SVD initialization is intuitive but does not by itself resolve the sparse initialization bottleneck in the regime considered here.

	\subsection{Two-View Phase Transition}
	\label{sec:exp-synergistic-gain}

	We next examine the predicted phase transition associated with the coupled profile term \(s_\bfu(\cdot)s_\bfv(\cdot)\). We set \(k_\bfu=k_\bfv=20\) and \(m=1000\), and vary the power-law decay rates \(\alpha_\bfu\) and \(\alpha_\bfv\) independently on the grid \(\{0,0.1,\ldots,2\}\). For each pair of rates, we report the mean of \(\max\{\sin\angle(\hat\bfu,\bfu),\sin\angle(\hat\bfv,\bfv)\}\) for Bi-SEP.

	\begin{figure}[htbp]
		\centering
		\includegraphics[width=0.85\linewidth]{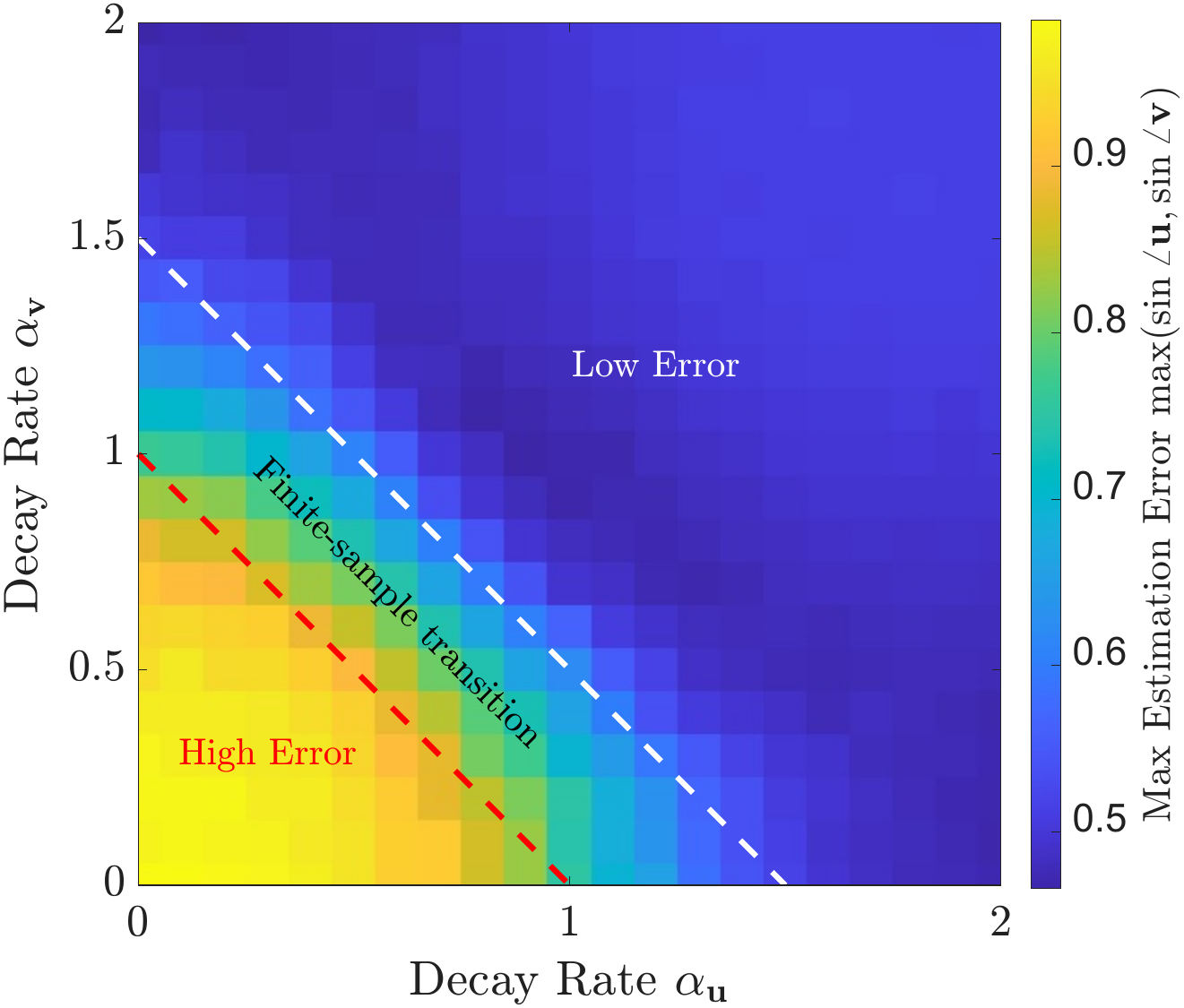}
		\caption{Empirical phase transition of Bi-SEP under varying power-law decay rates. The red dashed line denotes the boundary \(\alpha_\bfu+\alpha_\bfv=1\) predicted by the sufficient scaling law for \(m\), and the white dashed line indicates the empirical transition at this finite sample size.}
		\label{fig:exp-phase-transition}
	\end{figure}

	\Cref{fig:exp-phase-transition} shows a diagonal transition pattern: recovery improves as the aggregate decay \(\alpha_\bfu+\alpha_\bfv\) increases, even when one of the two views is only weakly structured. The empirical transition is not expected to coincide exactly with the asymptotic boundary because the experiment is finite-sample and uses fixed constants; nevertheless, its orientation is consistent with the theoretical condition. This supports the interpretation that energy concentration in one view can help the recovery of its partner view through the cross-covariance update.

	\subsection{Finite-Sample Behavior Under Unbalanced Sparsity}
	\label{sec:exp-decoupled-error}

	The third experiment examines error behavior in an unbalanced setting. We fix \(\bfu\) to be flat with \(k_\bfu=10\), let \(\bfv\) follow a steep power-law profile with \(\alpha_\bfv=3.0\), and vary \(k_\bfv\) from \(10\) to \(50\) with \(m=1000\). \Cref{fig:exp3-decoupled-alpha3} reports the errors of both views.

	\begin{figure}[htbp]
		\centering
		\subfloat[Error of \(\mathbf{u}\) (\(k_\bfu=10\) fixed)\label{fig:dec-a3-u}]{\includegraphics[width=0.48\linewidth]{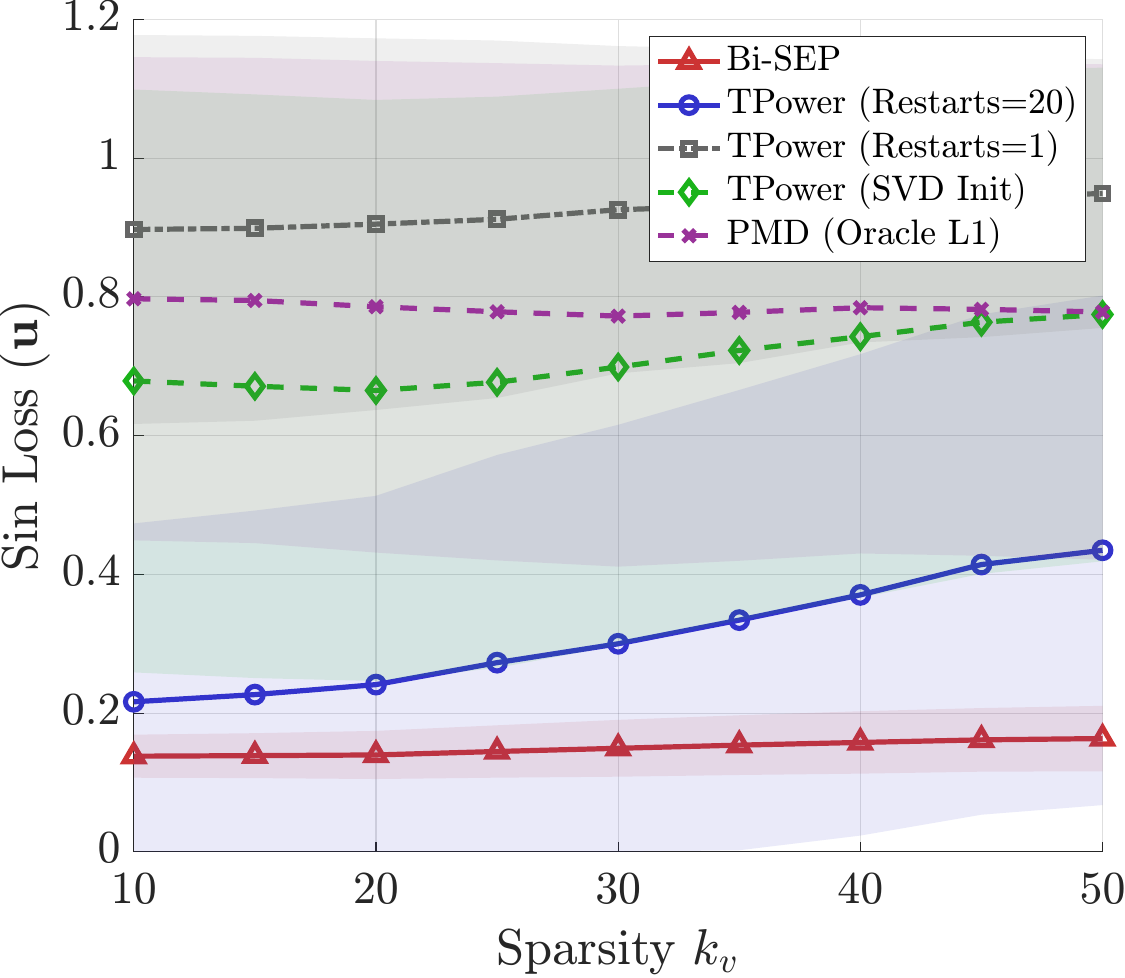}}
		\hfill
		\subfloat[Error of \(\mathbf{v}\) (\(k_\bfv\) varying)\label{fig:dec-a3-v}]{\includegraphics[width=0.48\linewidth]{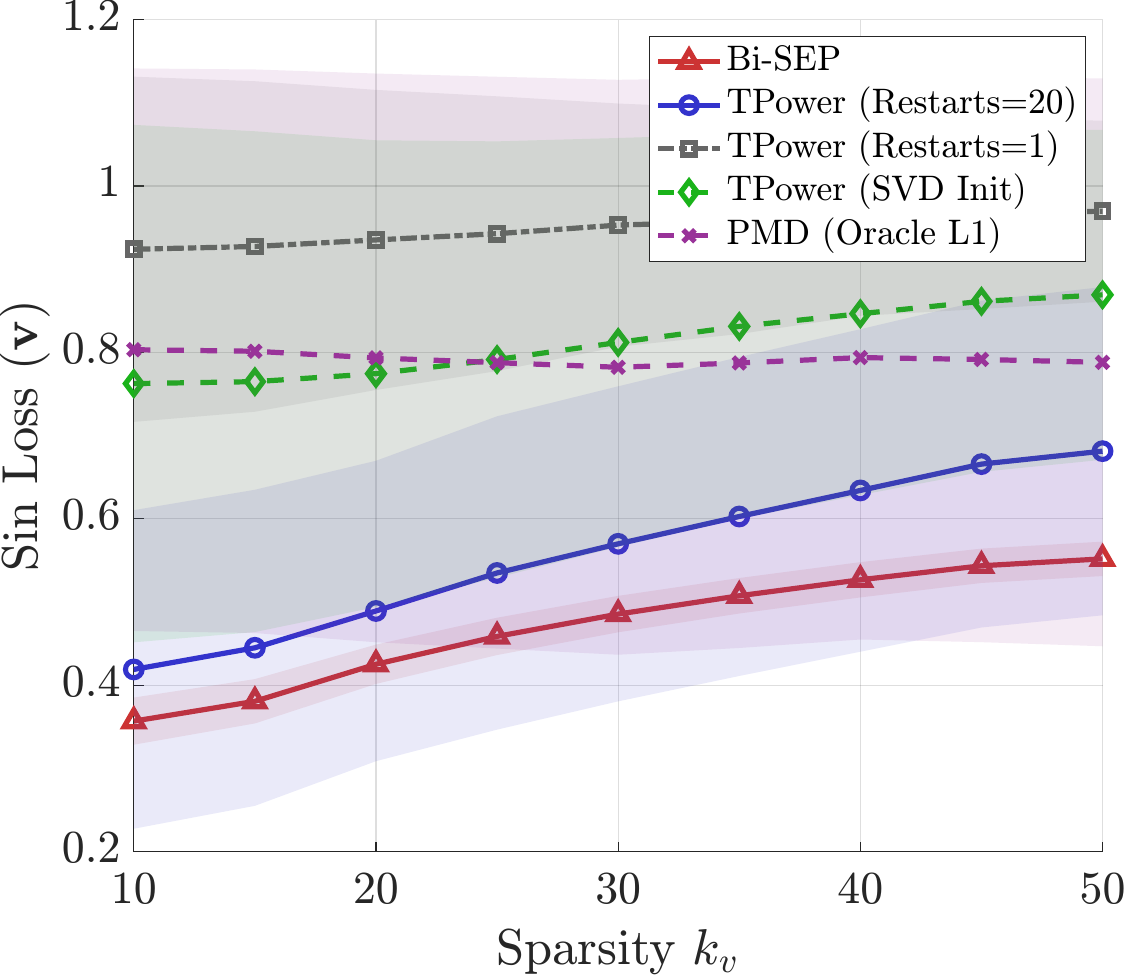}}
		\caption{Finite-sample errors under an unbalanced profile with flat \(\bfu\) and steeply decaying \(\bfv\). Over the displayed range, the error for \(\bfv\) increases with \(k_\bfv\), whereas the error for the fixed view \(\bfu\) changes comparatively little for Bi-SEP.}
		\label{fig:exp3-decoupled-alpha3}
	\end{figure}

	The error of the varying view \(\bfv\) increases with \(k_\bfv\), reflecting the larger support to be recovered. For the fixed view \(\bfu\), the mean error of Bi-SEP remains close to \(0.15\) across the displayed range, whereas the error of TPower-20 increases from approximately \(0.23\) to \(0.46\). Thus, over the displayed range, the fixed-view error of Bi-SEP changes little as the partner sparsity grows. This finite-sample behavior suggests a weaker dependence on the partner sparsity than the uniform bound captures; the current theory remains coupled, as discussed in \Cref{sec:discussion-coupling}.

	\subsection{Sensitivity to Sparsity Budget Misspecification}
	\label{sec:exp-sparsity-sensitivity}

	Given sparsity budgets \(k_\bfu\) and \(k_\bfv\), Bi-SEP adapts its support expansion to the signal energy profiles. We evaluate its sensitivity to these budgets in an unbalanced setting. The true sparsities are \(k_\bfu=10\) and \(k_\bfv=40\), while the supplied budgets are \((\operatorname{round}(\beta k_\bfu),\operatorname{round}(\beta k_\bfv))\), with \(\beta\in\{0.5,0.75,1.0,1.25,1.5,2.0\}\). Both views follow power-law profiles with \(\alpha=1.0\), and we set \(m=1000\).

	\begin{figure}[htbp]
		\centering
		\subfloat[Error of \(\mathbf{u}\) (\(k_\bfu=10\))\label{fig:sens-unbal-u}]{\includegraphics[width=0.48\linewidth]{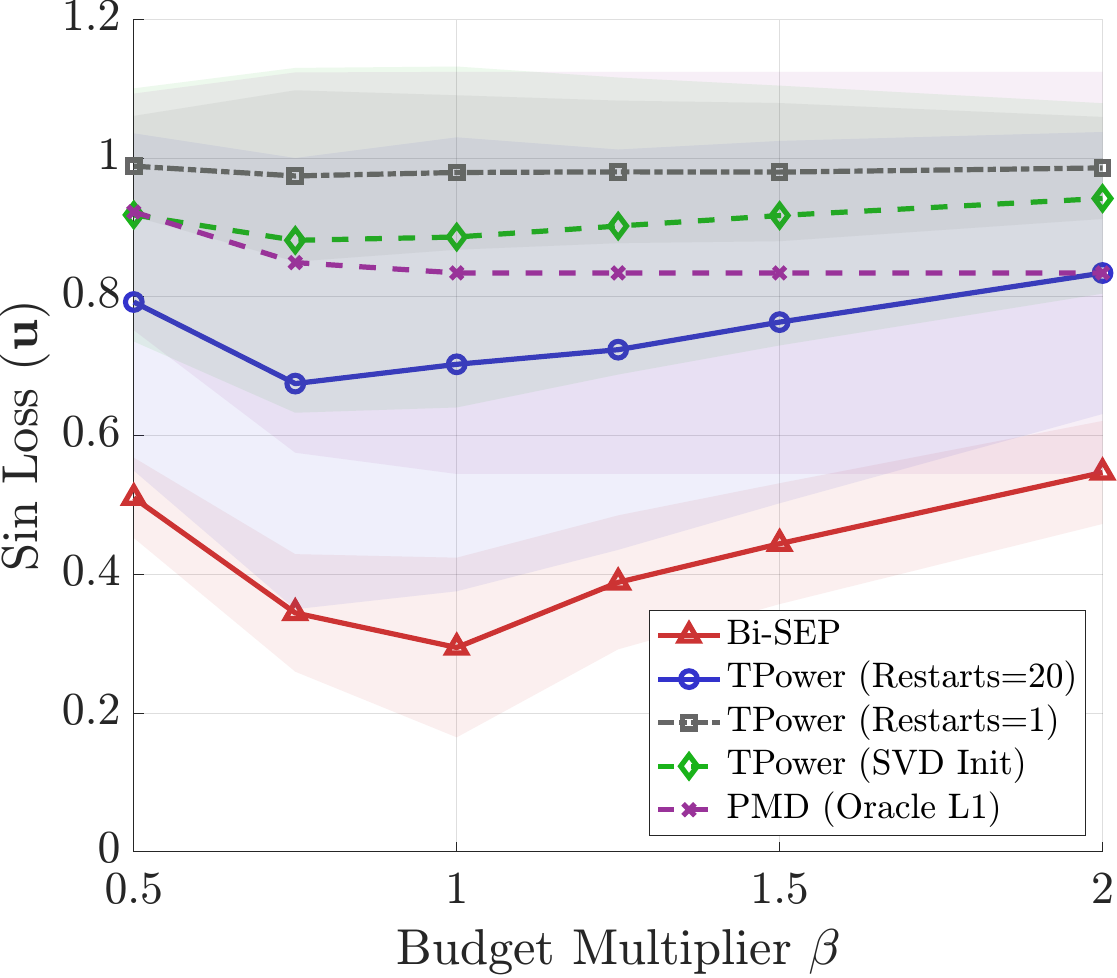}}
		\hfill
		\subfloat[Error of \(\mathbf{v}\) (\(k_\bfv=40\))\label{fig:sens-unbal-v}]{\includegraphics[width=0.48\linewidth]{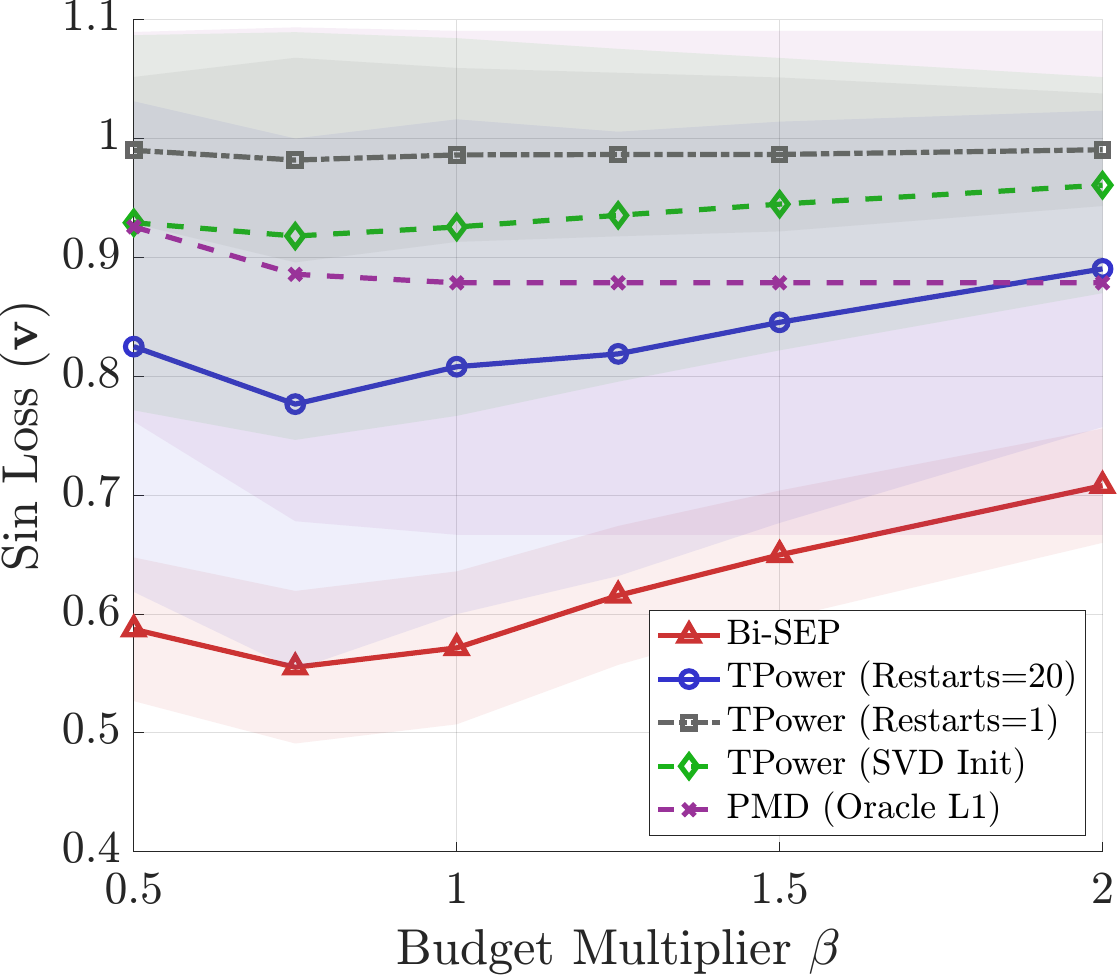}}
		\caption{Sensitivity to sparsity budget misspecification in the unbalanced setting. The mean error changes moderately for moderate misspecification but increases when the budgets are substantially overestimated.}
		\label{fig:sens-unbal}
	\end{figure}

	\Cref{fig:sens-unbal} shows that the error need not vary monotonically with the supplied budgets. Small budgets may omit weak signal coordinates, particularly in the view with the larger true support, whereas large budgets may include additional noise coordinates. These effects compete in finite samples, so moderate underestimation can sometimes yield lower error than the true budget. For large \(\beta\), however, the inclusion of noise coordinates generally degrades the estimates.

	In practice, one may select \((k_\bfu,k_\bfv)\) from a finite candidate grid. Directions estimated on a training split can be evaluated by their canonical correlation on validation data or by cross-validation. Alternatively, an information criterion can balance the empirical cross-covariance fit against a complexity penalty increasing with \(k_\bfu+k_\bfv\). Stability selection can repeatedly resample the observations and favor budgets whose estimated supports or directions remain stable. The two budgets should be selected jointly while allowing different values in the two views.

	\subsection{Sensitivity to Canonical Correlation}
	\label{sec:exp-rho}

	We next vary the canonical correlation while keeping $n=m=1000$ and $k_\bfu=k_\bfv=20$. We consider symmetric flat, power-law ($\alpha=1$), and exponential profiles. \Cref{tab:rho-sensitivity} reports representative correlation levels from the experiment. The metric is the maximum sine loss across the two directions, and each setting uses $1000$ Monte Carlo trials.

	\begin{table}[t]
		\caption{Sensitivity of the maximum two-view sine loss to the canonical correlation $\rho$. TP-20 denotes TPower-20.}
		\label{tab:rho-sensitivity}
		\centering
		\scriptsize
		\setlength{\tabcolsep}{2.4pt}
		\renewcommand{\arraystretch}{1.08}
		\begin{tabular}{@{}c cc cc cc@{}}
			\toprule
			& \multicolumn{2}{c}{\textsc{Flat}}
			& \multicolumn{2}{c}{\textsc{Power-Law}}
			& \multicolumn{2}{c}{\textsc{Exponential}} \\
			\cmidrule(lr){2-3}\cmidrule(lr){4-5}\cmidrule(l){6-7}
			$\rho$ & \textsc{TP-20} & \textsc{Bi-SEP} & \textsc{TP-20} & \textsc{Bi-SEP} & \textsc{TP-20} & \textsc{Bi-SEP} \\
			\midrule
			0.3 & 0.999 & 0.999 & 0.999 & 0.998 & 0.997 & 0.882 \\
			0.5 & 0.999 & 0.999 & 0.998 & 0.954 & 0.849 & 0.682 \\
			0.7 & 0.996 & 0.997 & 0.902 & 0.599 & 0.596 & 0.531 \\
			0.9 & 0.902 & 0.945 & 0.518 & 0.350 & 0.448 & 0.424 \\
			\bottomrule
		\end{tabular}
	\end{table}

	At $\rho=0.3$, the losses remain high for both methods, consistent with weak cross-view signal at the fixed sample size. The errors generally decrease as $\rho$ increases. Under the two concentrated profiles, Bi-SEP has lower error than TPower-20 at the displayed correlation levels, whereas their errors are similar in the flat regime except that TPower-20 is lower at $\rho=0.9$. The degradation at small $\rho$ is consistent with the correlation dependence in \Cref{thm:recovery_guarantee}.

	\subsection{Performance under Non-Gaussian Sampling}
	\label{sec:exp-nongaussian}

	We next examine performance under non-Gaussian sampling. We use the same power-law profile as in \Cref{fig:ablation-power}, with \(k_\bfu=k_\bfv=20\), and consider Gaussian, Rademacher \(\pm1\), Laplace, and Student's \(t\) innovations with five degrees of freedom. For all four distributions, the population model has identity marginal covariances and cross-covariance \(\Sigma_{xy}=\rho\bfu\bfv^\top\); the data are centered before the sample cross-covariance is formed. The distributions differ in their tail behavior: Gaussian and Rademacher innovations are sub-Gaussian, Laplace innovations are sub-exponential, and Student's \(t\) innovations have polynomial tails. The algorithms and their tuning are unchanged.

	\begin{figure*}[t]
		\centering
		\subfloat[Gaussian\label{fig:ng-gauss}]{\includegraphics[width=0.24\linewidth]{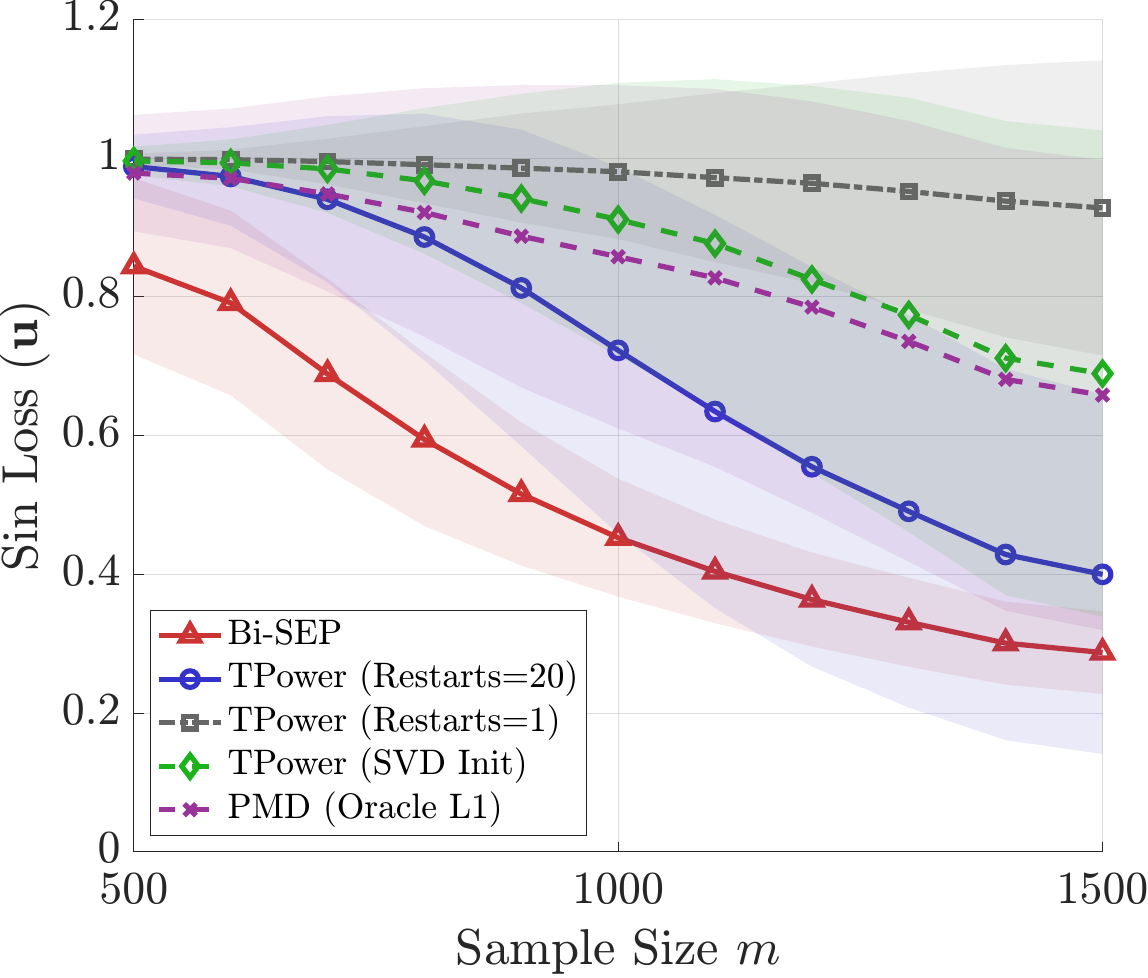}}
		\hfill
		\subfloat[Rademacher \(\pm1\) innovations\label{fig:ng-rad}]{\includegraphics[width=0.24\linewidth]{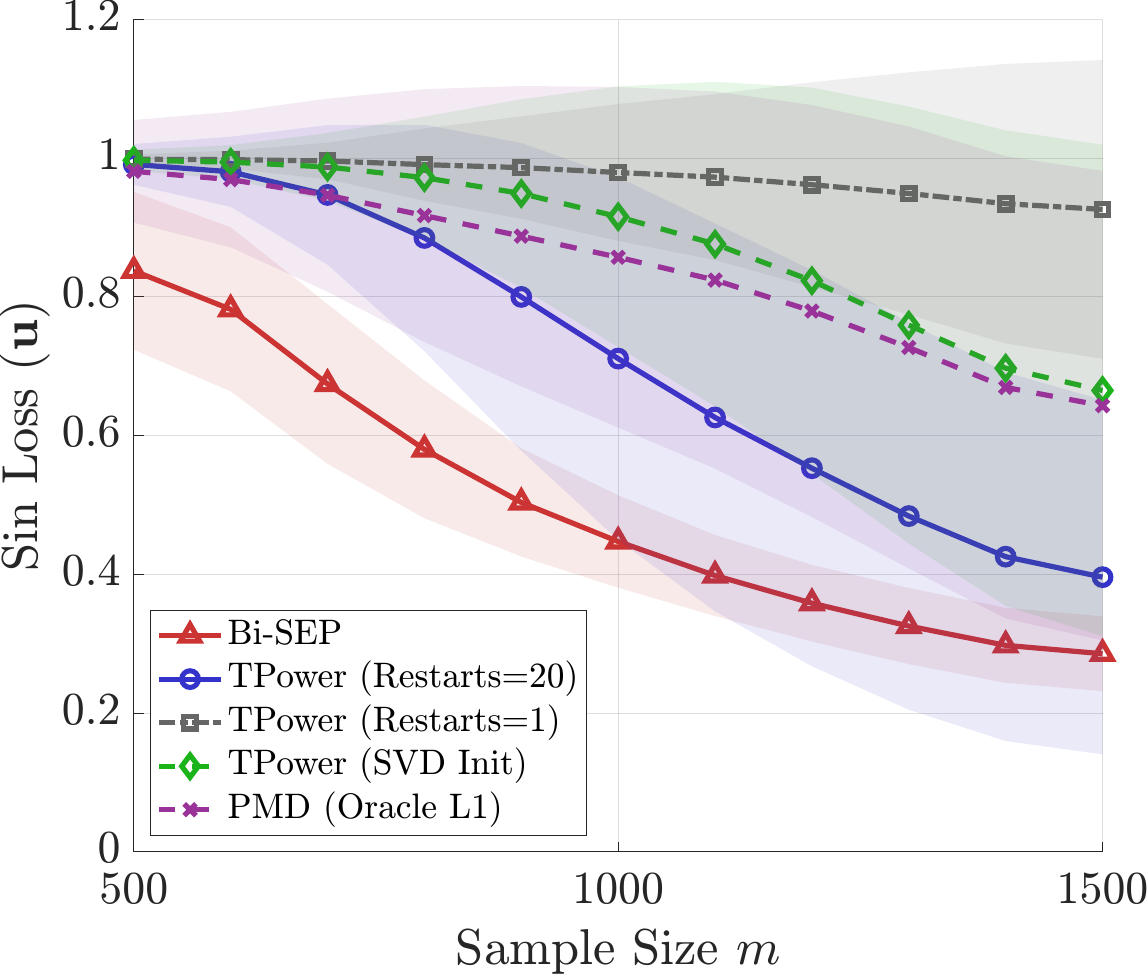}}
		\hfill
		\subfloat[Laplace\label{fig:ng-lap}]{\includegraphics[width=0.24\linewidth]{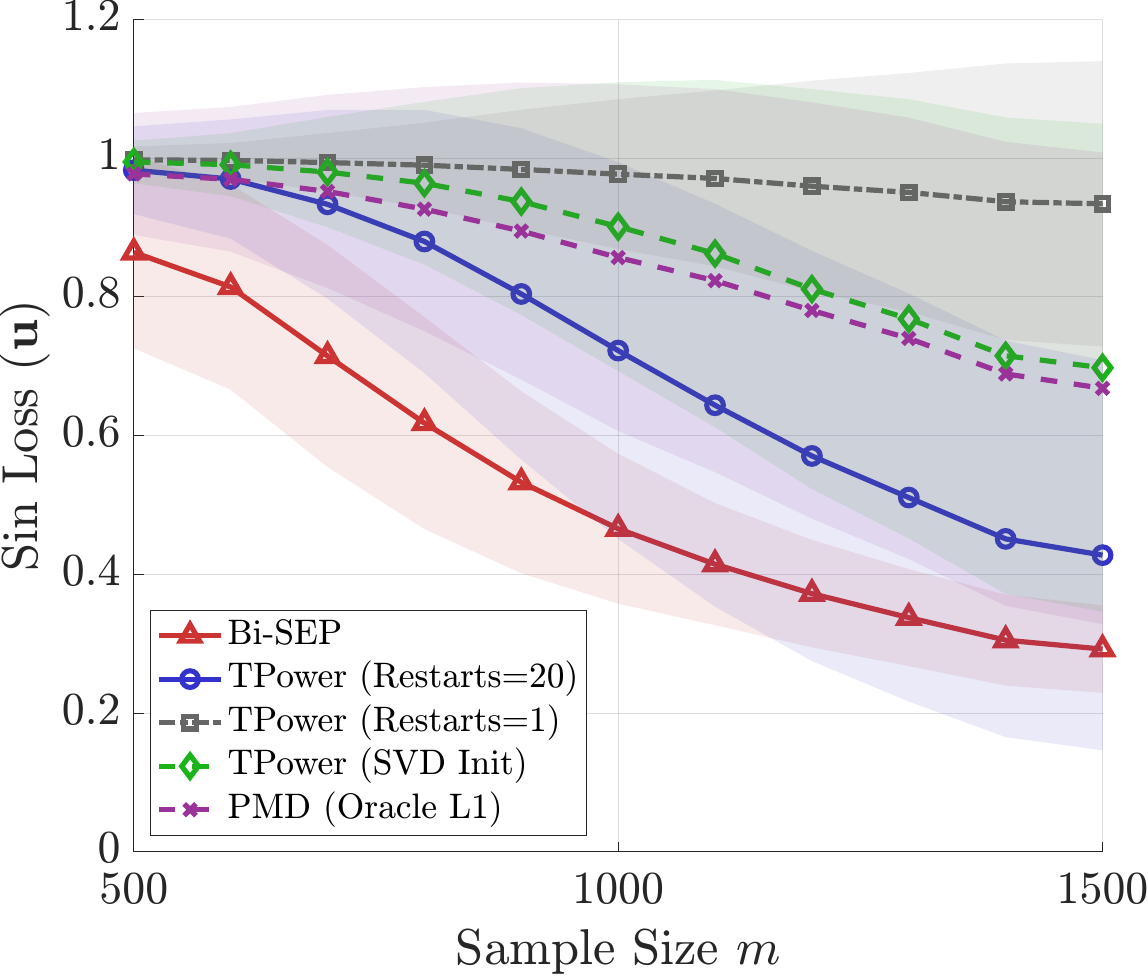}}
		\hfill
		\subfloat[Student's \(t\) (\(df=5\))\label{fig:ng-t}]{\includegraphics[width=0.24\linewidth]{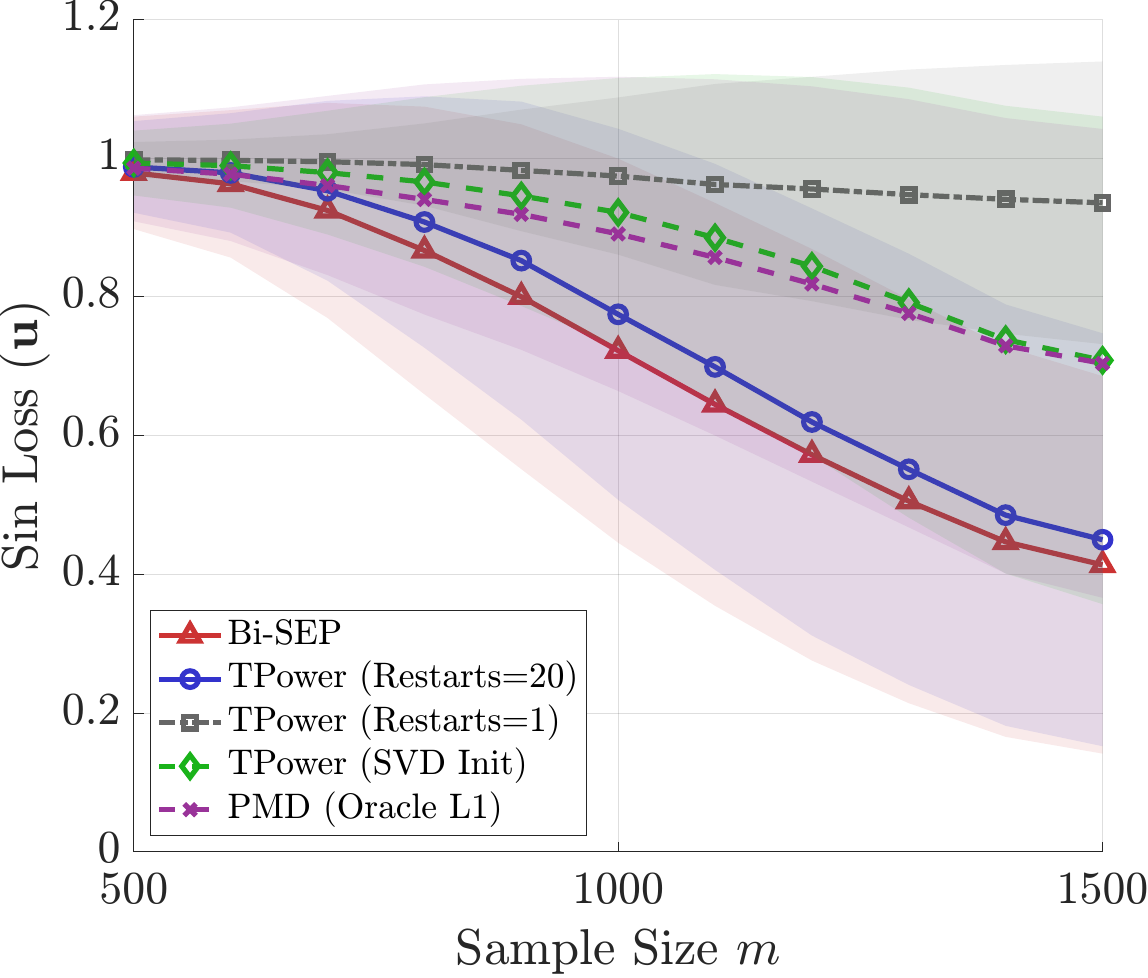}}
		\caption{Estimation error under different sampling distributions in the power-law profile setting. The distributions share the same population covariance structure and cross-covariance but differ in tail behavior.}
		\label{fig:nongaussian-robust}
	\end{figure*}

	As shown in \Cref{fig:nongaussian-robust}, the Gaussian and Rademacher cases have similar error curves. The error is higher for Laplace sampling and increases further for Student's \(t\) sampling. This ordering is consistent with the different tail behaviors of the four distributions. This experiment provides an empirical comparison. A theoretical extension to these distributions would require corresponding concentration results for the sample cross-covariance.

	\subsection{Spatial and Spectral Source Localization}
	\label{sec:exp-application}

	The final experiment considers a controlled model motivated by spatial and spectral source localization. The first view represents a spatial beamspace vector with \(k_\bfu=10\) and a concentrated power-law profile \(\alpha_\bfu=3.0\), which reflects directional beam energy concentration. The second view represents a diffuse spectral modal vector with \(k_\bfv=25\) and a flat profile, which reflects a wideband spectral distribution. The leading cross-covariance component links the concentrated spatial response to the broader spectral signature. The sample size varies from \(m=500\) to \(1500\).

	\begin{figure}[htbp]
		\centering
		\subfloat[Spatial view\label{fig:app-u}]{\includegraphics[width=0.48\linewidth]{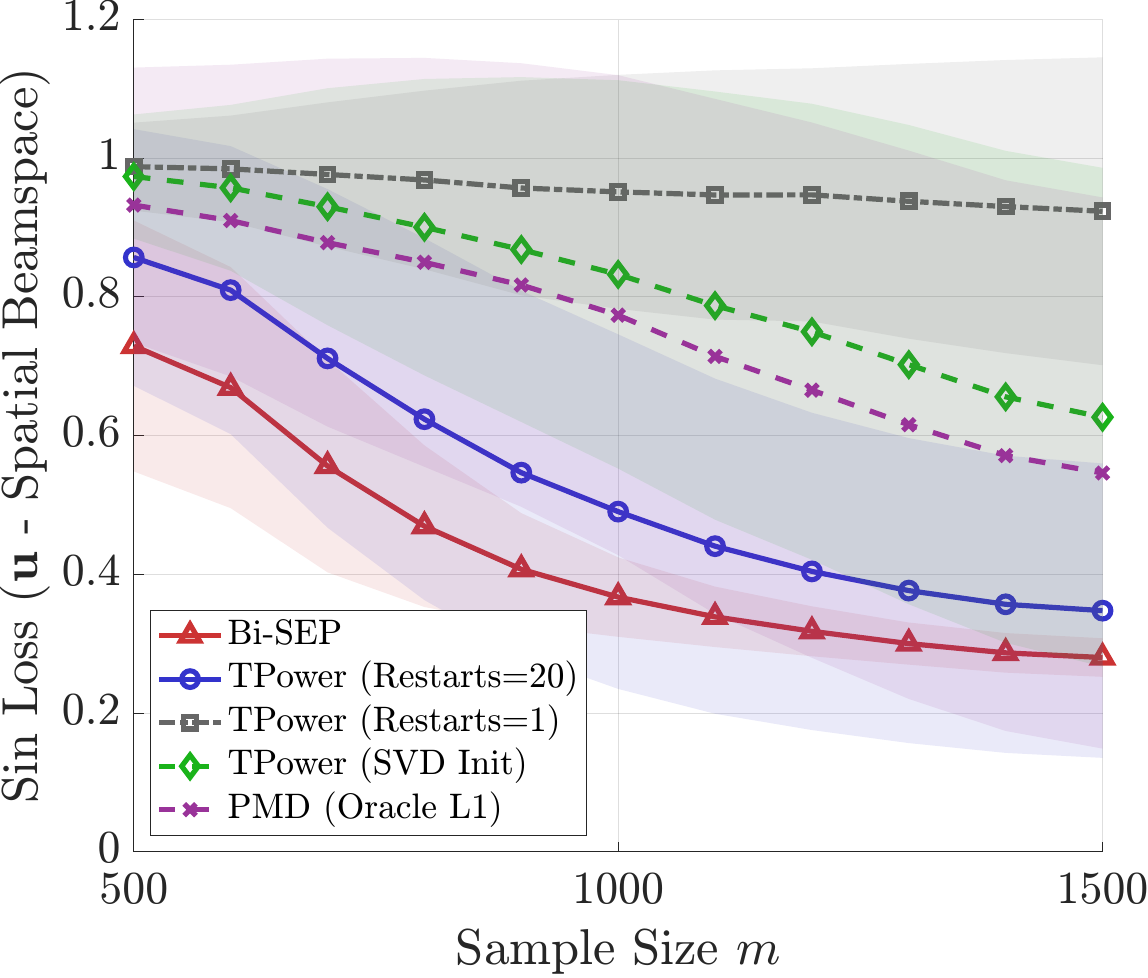}}
		\hfill
		\subfloat[Spectral view\label{fig:app-v}]{\includegraphics[width=0.48\linewidth]{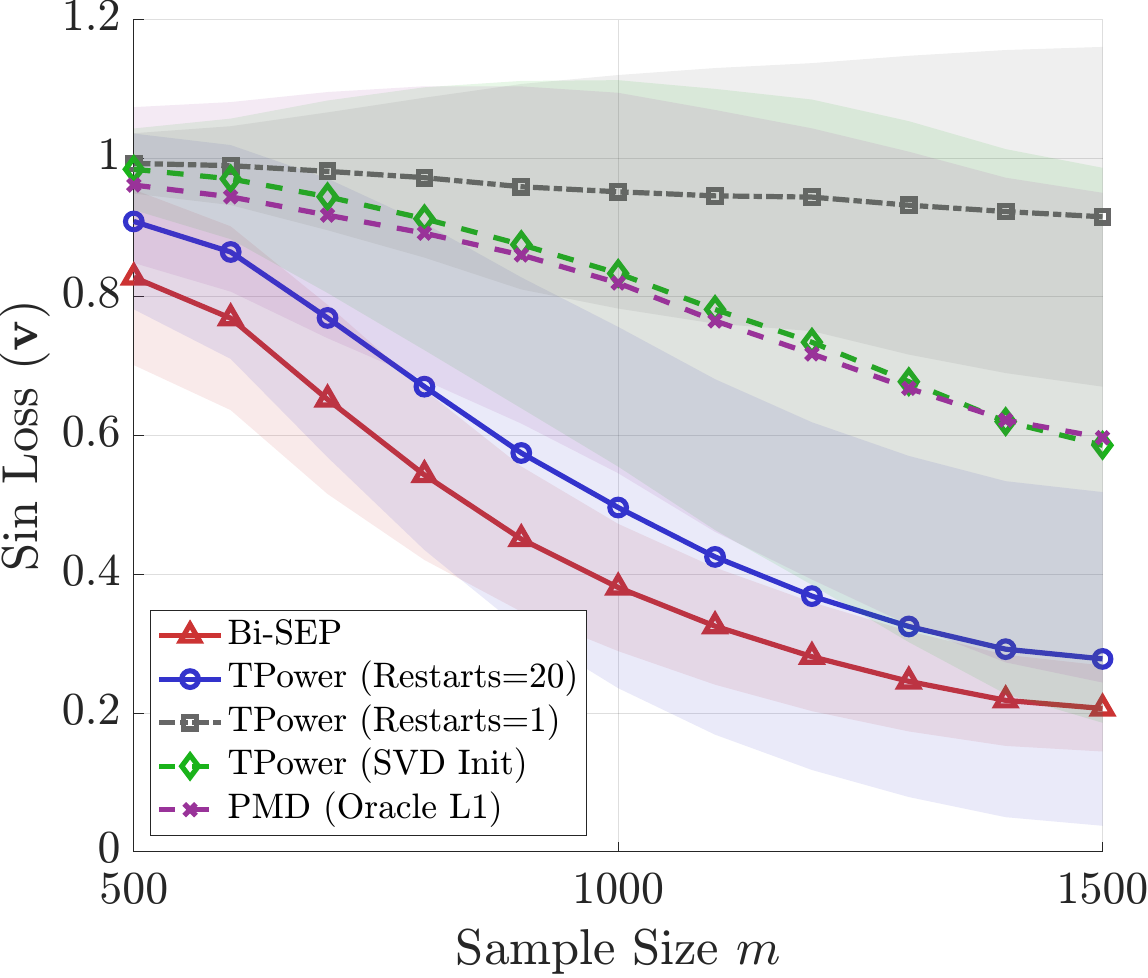}}
		\caption{Simulation motivated by spatial and spectral source localization, with \(k_\bfu=10\), \(k_\bfv=25\), \(\alpha_\bfu=3.0\), and \(\alpha_\bfv=0\).}
		\label{fig:exp6-app}
	\end{figure}

	\Cref{fig:exp6-app} shows that Bi-SEP has the lowest mean error over most of the displayed sample range. At \(m=1500\), its mean errors are approximately \(0.27\) for the spatial vector and \(0.19\) for the spectral vector, compared with approximately \(0.34\) and \(0.26\) for TPower-20. TPower-SVD and PMD also improve as \(m\) increases, but remain above Bi-SEP over the plotted range. This result is consistent with the proposed stagewise mechanism in an asymmetric setting in which one view has concentrated energy and the other has a diffuse profile.

%% file: 07Conclusion.tex
\section{Conclusion}
\label{sec:conclusion}

In this paper, we investigated Sparse Canonical Correlation Analysis (SCCA) beyond the standard worst case of flat signals by explicitly modeling the nonuniform energy distribution of signals from multiple views. We introduced Bilateral Spectral Energy Pursuit (Bi-SEP) to adaptively construct the support sets and track cross-view signal energy. By establishing a profile-adaptive sample complexity, we demonstrated that the statistical-to-computational gap is intrinsically tied to adversarial flat signals and can be naturally bypassed. Specifically, signals with highly concentrated energy require substantially fewer samples for reliable recovery.

This profile-adaptive framework provides clear directions for future research. An important extension is to multiple canonical components, where component interference prevents the rank-one initialization argument from applying directly; possible routes include sequential extraction or sparse-subspace formulations. While our analysis establishes achievable sample complexity upper bounds for continuous decay profiles, defining the corresponding information-theoretic minimax lower bounds remains an open theoretical challenge. Furthermore, the algorithmic principle of stagewise energy pursuit can be investigated for other multivariate statistical models where similar worst-case initialization bottlenecks limit practical performance. Domain-validated real-data studies are another important direction for future work.

%% file: 08Appendices.tex
\appendices

\section{Basic Lemmas}\label{sec:app-A}

This appendix provides essential auxiliary results regarding the perturbation of singular subspaces. Specifically, we present a tailored, rank-$1$ specialization of Wedin's $\sin\Theta$ theorem, which serves as the core mathematical tool for analyzing the iterative alignment in our framework.

\begin{lemma}[Wedin's $\sin\Theta$ theorem for rank-1 matrices {\cite{wedin1972perturbation, stewart1990matrix}}]
	\label{lemma:wedin_rank1}
	Let $\bfA \in \mathbb{R}^{n_1 \times n_2}$ be a rank-$1$ matrix with singular value decomposition $\bfA = \sigma_1 \bfu\bfv^\top$, where $\sigma_1 > 0$ and $\|\bfu\|_2 = \|\bfv\|_2 = 1$. Consider the perturbed matrix $\tilde{\bfA} = \bfA + \bfE$. Let $(\tilde\sigma_1, \tilde\bfu, \tilde\bfv)$ denote the leading singular triplet of $\tilde{\bfA}$, and let $\tilde\sigma_2$ be its second largest singular value. Define the spectral gap:
	\begin{equation}
		\delta := \min\big\{\tilde\sigma_1,\, \sigma_1 - \tilde\sigma_2\big\}.
	\end{equation}
	If $\delta > 0$, then the principal angles between the unperturbed and perturbed singular vectors satisfy
	\begin{equation}
		\max \big\{ \sin\angle(\tilde\bfu, \bfu),\, \sin\angle(\tilde\bfv, \bfv) \big\} \le \frac{\|\bfE\|_2}{\delta}.
	\end{equation}
\end{lemma}

\begin{corollary}[Rank-1 specialization via Weyl's inequality]
	\label{cor:wedin_rank1_weyl}
	Under the setting of \Cref{lemma:wedin_rank1}, if the spectral norm of the perturbation satisfies $\|\bfE\|_2 < \sigma_1$, then the singular subspace perturbation is bounded by:
	\begin{equation}
		\max\big\{ \sin \angle(\tilde{\bfu}, \bfu),\, \sin \angle(\tilde{\bfv}, \bfv) \big\} \le \frac{\|\bfE\|_2}{\sigma_1 - \|\bfE\|_2}.
	\end{equation}
	Furthermore, in the high signal-to-noise regime where $\|\bfE\|_2 \le \sigma_1/2$ (a condition strictly guaranteed under our sample complexity assumptions), the bound simplifies to
	\begin{equation}
		\max\big\{ \sin \angle(\tilde{\bfu}, \bfu),\, \sin \angle(\tilde{\bfv}, \bfv) \big\} \le \frac{2\|\bfE\|_2}{\sigma_1}.
	\end{equation}
\end{corollary}

\begin{proof}[Proof of \Cref{cor:wedin_rank1_weyl}]
	Since $\sigma_2(\bfA) = 0$, Weyl's inequality gives $\tilde{\sigma}_1 \ge \sigma_1 - \|\bfE\|_2$ and $\tilde{\sigma}_2 \le \|\bfE\|_2$. Thus, the spectral gap satisfies $\delta \ge \sigma_1 - \|\bfE\|_2 > 0$. Substituting this into \Cref{lemma:wedin_rank1} yields the first inequality. The simplified bound directly follows by noting that $\sigma_1 - \|\bfE\|_2 \ge \sigma_1/2$ under the given assumption.
\end{proof}

\begin{lemma}[One-step hard-thresholded refinement]\label{lem:threshold-refinement}
		Let $\bfx\in\mathbb R^n$ be a unit vector with $\|\bfx\|_0\le k$, and let $\bfy=a\bfx+\bfe$ with $a>0$. Suppose $\|\bfe_A\|_2\le b$ for every $A\subset[n]$ with $|A|\le k$. If $a>2b$, then $\bfh=\mathcal H_k(\bfy)/\|\mathcal H_k(\bfy)\|_2$ is well defined and satisfies
		\begin{equation}\label{eq:threshold-refinement-lemma}
			\cos\angle(\bfh,\bfx)\ge\frac{a-2b}{a+b},
			\qquad
			\sin\angle(\bfh,\bfx)\le\frac{3b}{a-2b}.
		\end{equation}
	\end{lemma}

	\begin{proof}
		Let $S=\operatorname{supp}(\bfx)$ and $K=\operatorname{supp}(\mathcal H_k(\bfy))$, so $\mathcal H_k(\bfy)=\bfy_K$. The top-$k$ property gives $\|\bfy_{S\setminus K}\|_2\le\|\bfy_{K\setminus S}\|_2=\|\bfe_{K\setminus S}\|_2\le b$. Consequently,
		\begin{align}
			\langle\mathcal H_k(\bfy),\bfx\rangle
			 & =a+\langle\bfe,\bfx\rangle-\langle\bfy_{S\setminus K},\bfx\rangle
			\ge a-2b,\label{eq:threshold-inner}                                  \\
			\|\mathcal H_k(\bfy)\|_2
			 & \le a\|\bfx_K\|_2+\|\bfe_K\|_2\le a+b.\label{eq:threshold-norm}
		\end{align}
		The first inequality is positive and, together with \eqref{eq:threshold-norm}, proves the cosine bound. Moreover,
		\begin{align}
			\nonumber \|(\bfI-\bfx\bfx^\top)\mathcal H_k(\bfy)\|_2
			 & \le\|\mathcal H_k(\bfy)-a\bfx\|_2                \\
			 & \le\|\bfe_K\|_2+\|a\bfx_{S\setminus K}\|_2\le3b,
		\end{align}
		where the last inequality follows from $a\bfx_{S\setminus K}=\bfy_{S\setminus K}-\bfe_{S\setminus K}$. Finally, \eqref{eq:threshold-inner} also gives $\|\mathcal H_k(\bfy)\|_2\ge a-2b$. Dividing the orthogonal-component bound by this lower bound proves the sine inequality.
	\end{proof}

\section{Proofs of Propositions}
This appendix provides the detailed proofs for the auxiliary propositions introduced in the main text to support the algorithmic induction of Bi-SEP.
\subsection{Proof of \Cref{prop:union_event}}

The proof employs a standard $\epsilon$-net argument combined with Bernstein's inequality and a union bound over all possible sparse supports; see, e.g., standard concentration treatments in~\cite{Boucheron2013Concentration,vershynin2012introduction}.
Fix any supports $S_1, S_2 \subset [n]$ with $|S_1|=s_1$ and $|S_2|=s_2$. We can write the empirical cross-covariance block as
\begin{equation}
	\bfW_{S_1,S_2} = \frac{1}{m} \sum_{i=1}^m \Big( \bfx_{i,S_1}\bfy_{i,S_2}^\top - \bbE[\bfx_{S_1}\bfy_{S_2}^\top] \Big) =: \frac{1}{m} \sum_{i=1}^m \bfZ_i,
\end{equation}
where $\bfZ_i \in \bbR^{s_1\times s_2}$ are i.i.d. zero-mean random matrices.

\paragraph{Step 1: Concentration for fixed unit vectors}
Let $\bfa \in \mathbb{S}^{s_1-1}$ and $\bfb \in \mathbb{S}^{s_2-1}$. Define the centered scalar random variables:
\begin{align}
	\nonumber \xi_i(\bfa,\bfb) & := \bfa^\top \bfZ_i \bfb                                                                                                         \\
	                           & = \big(\bfa^\top \bfx_{i,S_1}\big)\big(\bfb^\top \bfy_{i,S_2}\big) - \bbE\big[(\bfa^\top \bfx_{S_1})(\bfb^\top \bfy_{S_2})\big].
\end{align}
Since $\bfx$ and $\bfy$ are jointly Gaussian, the projections $\bfa^\top \bfx_{S_1}$ and $\bfb^\top \bfy_{S_2}$ are sub-Gaussian random variables with uniformly bounded $\psi_2$-norms. Consequently, their product is sub-exponential. Using the standard property $\|XY\|_{\psi_1} \lesssim \|X\|_{\psi_2}\|Y\|_{\psi_2}$, there exists an absolute constant $K$ (depending on $\rho$) such that $\|\xi_i(\bfa,\bfb)\|_{\psi_1} \le K$ uniformly over all valid configurations.

Applying Bernstein's inequality for sub-exponential variables, there exist universal constants $c_1, C_1 > 0$ such that for any $u \ge 0$, it holds that
\begin{equation}
	\bbP\!\left( \left|\frac{1}{m}\sum_{i=1}^m \xi_i(\bfa,\bfb)\right| \ge C_1\left(\sqrt{\frac{u}{m}} + \frac{u}{m}\right) \right) \le 2e^{-c_1 u}.
	\label{eq:prop_union_event_bernstein}
\end{equation}

\paragraph{Step 2: $\epsilon$-net argument for fixed supports}
Let $\mcN_{s_1}$ and $\mcN_{s_2}$ denote $1/4$-nets of the unit spheres $\mathbb{S}^{s_1-1}$ and $\mathbb{S}^{s_2-1}$, respectively. By standard volumetric bounds, we can choose these nets such that $|\mcN_{s_1}| \le 9^{s_1}$ and $|\mcN_{s_2}| \le 9^{s_2}$. A standard successive approximation argument yields
\begin{align}
	\nonumber \|\bfW_{S_1,S_2}\|_2 & = \sup_{\bfa, \bfb} \bfa^\top \bfW_{S_1,S_2} \bfb                                                              \\
	                               & \le 2\sup_{\bfa\in\mcN_{s_1}, \bfb\in\mcN_{s_2}} \bfa^\top \bfW_{S_1,S_2} \bfb.\label{eq:prop_union_event_sup}
\end{align}
Taking a union bound of \eqref{eq:prop_union_event_bernstein} over the nets $\mcN_{s_1} \times \mcN_{s_2}$ gives
\begin{align}
	\nonumber & \bbP\!\left( \|\bfW_{S_1,S_2}\|_2 \ge C_2\left(\sqrt{\frac{u}{m}} + \frac{u}{m}\right) \right) \\
	          & \qquad \le 2 \exp\!\big((s_1+s_2)\log 9 - c_1 u\big).
	\label{eq:prop_union_event_union}
\end{align}

\paragraph{Step 3: Global union bound over all supports}
The number of possible support pairs $(S_1, S_2)$ of fixed sizes $s_1$ and $s_2$ is bounded by $\binom{n}{s_1}\binom{n}{s_2} \le (\frac{en}{s_1})^{s_1} (\frac{en}{s_2})^{s_2}$. Applying a union bound to \eqref{eq:prop_union_event_union} over all such pairs, the probability that the spectral norm exceeds the threshold for \textit{any} pair of sizes $(s_1, s_2)$ is bounded by
\begin{equation}
	2 \exp\!\left( s_1\log\frac{en}{s_1} + s_2\log\frac{en}{s_2} + (s_1+s_2)\log 9 - c_1 u \right).
\end{equation}
Setting $u := \frac{2}{c_1}\left( s_1\log\frac{en}{s_1} + s_2\log\frac{en}{s_2} + (s_1+s_2)\log 9 + t \right)$, the right-hand side reduces to $2e^{-2t}$. Thus, with probability at least $1-2e^{-2t}$, uniformly for all supports of sizes exactly $(s_1,s_2)$, it holds that
\begin{equation}
	\|\bfW_{S_1,S_2}\|_2 \le C_3\left( \sqrt{\frac{\psi(s_1,s_2) + t}{m}} + \frac{\psi(s_1,s_2) + t}{m} \right),
\end{equation}
where $\psi(s_1,s_2) := s_1\log\frac{en}{s_1} + s_2\log\frac{en}{s_2}$, and the constant $\log 9$ terms have been absorbed into $C_3$.

Finally, applying a union bound over all sparsity levels $1 \le s_1 \le k_\bfu$ and $1 \le s_2 \le k_\bfv$ introduces an additional factor of $k_\bfu k_\bfv$, requiring a minor adjustment $t \to t + \log(k_\bfu k_\bfv)$. Since $\log\frac{en}{s} \le \log n + 1$, we trivially have $\psi(s_1, s_2) \lesssim (|S_1|+|S_2|)\log n$. Under the sample size regime $m \gtrsim (|S_1|+|S_2|)\log n + t$, the linear term is strictly dominated by the sub-Gaussian square-root term. Consequently, the bound simplifies to the purely square-root form and establishes the proposition. \hfill \IEEEQEDopen

\subsection{Proof of \Cref{prop:initialization}}

Let $(i^*, j^*)$ denote the indices of the largest entries of $\bfu$ and $\bfv$ in absolute value, respectively. By definition of the structure function, $|u_{i^*} v_{j^*}| = |u_{(1)} v_{(1)}| = 1/\sqrt{s_\bfu(1) s_\bfv(1)}$. The algorithm initializes by selecting $(i_0, j_0) = \arg\max_{i,j} |(\bhSigma_{xy})_{ij}|$, which inherently implies
\begin{equation}
	|(\bhSigma_{xy})_{i_0 j_0}| \ge |(\bhSigma_{xy})_{i^* j^*}|.
\end{equation}

Applying the triangle inequality and the reverse triangle inequality to the decomposition $\bhSigma_{xy} = \rho \bfu \bfv^\top + \bfW$, we expand both sides:
\begin{align}
	\nonumber \rho |u_{i_0} v_{j_0}| + |W_{i_0 j_0}| & \ge |(\bhSigma_{xy})_{i_0 j_0}|             \\
	\nonumber                                        & \ge |(\bhSigma_{xy})_{i^* j^*}|             \\
	                                                 & \ge \rho |u_{(1)} v_{(1)}| - |W_{i^* j^*}|.
\end{align}
Rearranging terms yields a lower bound for the signal product on the selected indices:
\begin{align}
	\nonumber |u_{i_0} v_{j_0}| & \ge |u_{(1)} v_{(1)}| - \frac{1}{\rho}\left( |W_{i_0 j_0}| + |W_{i^* j^*}| \right) \\
	                            & \ge \frac{1}{\sqrt{s_\bfu(1) s_\bfv(1)}} - \frac{2}{\rho} \max_{i,j} |W_{ij}|.
\end{align}

Specializing the uniform noise bound under event $\mathcal{E}$ to $1 \times 1$ submatrices, the entry-wise noise is uniformly bounded by $\max_{i,j} |W_{ij}| \le C(1+\rho)\sqrt{(\log n)/m}$. Substituting this into the inequality gives
\begin{equation}
	|u_{i_0} v_{j_0}| \ge \frac{1}{\sqrt{s_\bfu(1) s_\bfv(1)}} - \frac{2C(1+\rho)}{\rho} \sqrt{\frac{\log n}{m}}.
\end{equation}

By the initial sample complexity condition \eqref{eq:initial-t-0-sample}, we have $\sqrt{m} \ge \frac{2C(1+\rho)}{\rho(1-\sqrt{\gamma})} \sqrt{s_\bfu(1)s_\bfv(1)\log n}$. Substituting this condition, the subtraction term is strictly bounded by $(1-\sqrt{\gamma})/\sqrt{s_\bfu(1) s_\bfv(1)}$. Consequently, we have
\begin{equation}
	|u_{i_0} v_{j_0}| \ge \frac{1 - (1-\sqrt{\gamma})}{\sqrt{s_\bfu(1) s_\bfv(1)}} = \sqrt{\frac{\gamma}{s_\bfu(1) s_\bfv(1)}}.
\end{equation}
\hfill \IEEEQEDopen

\subsection{Proof of \Cref{prop:alignment}}
Consider the restricted empirical cross-covariance matrix on the given supports $S_\bfu$ and $S_\bfv$, denoted as $(\bhSigma_{xy})_{S_\bfu, S_\bfv} = \rho \bfu_{S_\bfu} \bfv_{S_\bfv}^\top + \bfW_{S_\bfu, S_\bfv}$. The unperturbed signal matrix $\rho \bfu_{S_\bfu} \bfv_{S_\bfv}^\top$ is strictly rank-1 with a unique non-zero singular value $\sigma_1 = \rho \|\bfu_{S_\bfu}\|_2 \|\bfv_{S_\bfv}\|_2$. Its corresponding left and right singular vectors are exactly the normalized true signal segments $\bar{\bfu} = \bfu_{S_\bfu}/\|\bfu_{S_\bfu}\|_2 \in \mathbb{R}^{|S_\bfu|}$ and $\bar{\bfv} = \bfv_{S_\bfv}/\|\bfv_{S_\bfv}\|_2 \in \mathbb{R}^{|S_\bfv|}$.

The estimators $\hat{\bfu}$ and $\hat{\bfv}$ (restricted to their supports $S_\bfu$ and $S_\bfv$) are, by algorithmic definition, the leading singular vectors of the perturbed matrix $(\bhSigma_{xy})_{S_\bfu, S_\bfv}$. Applying \Cref{cor:wedin_rank1_weyl} to this restricted block, we bound the principal angle between the perturbed and unperturbed singular vectors:
\begin{equation}
	\sin \angle (\hat{\bfu}_{S_\bfu}, \bar{\bfu}) \le \frac{2\|\bfW_{S_\bfu, S_\bfv}\|_2}{\rho \|\bfu_{S_\bfu}\|_2 \|\bfv_{S_\bfv}\|_2}.
\end{equation}

Since both $\hat{\bfu}_{S_\bfu}$ and $\bar{\bfu}$ are unit vectors in $\mathbb{R}^{|S_\bfu|}$, their inner product is exactly the cosine of their principal angle. Using the trigonometric identity $|\langle \hat{\bfu}_{S_\bfu}, \bar{\bfu} \rangle| = \cos \angle (\hat{\bfu}_{S_\bfu}, \bar{\bfu}) = \sqrt{1 - \sin^2 \angle (\hat{\bfu}_{S_\bfu}, \bar{\bfu})}$, we obtain
\begin{equation}
	|\langle \hat{\bfu}_{S_\bfu}, \bar{\bfu} \rangle| \ge \sqrt{1 - \left(\frac{2\|\bfW_{S_\bfu, S_\bfv}\|_2}{\rho \|\bfu_{S_\bfu}\|_2 \|\bfv_{S_\bfv}\|_2}\right)^2}.
\end{equation}

Crucially, as $\hat{\bfu}$ is supported on $S_\bfu$ and padded with zeros, the inner product with the global vector $\bfu$ satisfies $|\langle \hat{\bfu}, \bfu \rangle| = |\langle \hat{\bfu}_{S_\bfu}, \bfu_{S_\bfu} \rangle| = \|\bfu_{S_\bfu}\|_2 |\langle \hat{\bfu}_{S_\bfu}, \bar{\bfu} \rangle|$, which precisely yields the stated bound \eqref{eq:u-align}. The proof for $\hat{\bfv}$ follows identically by symmetry.
\hfill \IEEEQEDopen

\subsection{Proof of \Cref{prop:reselection}}

We focus on the update for $\bfu$. Let $T^*$ denote the theoretical support of the $|S_\bfu^{(t+1)}|$ largest entries of $\bfu$ in magnitude. By the definition of the structure function, the captured energy on this optimal set is $\|\bfu_{T^*}\|_2 = 1/\sqrt{s_\bfu(|S_\bfu^{(t+1)}|)}$.

The algorithm selects $S_\bfu^{(t+1)}$ as the indices corresponding to the largest $|S_\bfu^{(t+1)}|$ entries of the proxy vector $\bfr_\bfu = \bhSigma_{xy} \hat{\bfv}^{(t)}$. This index optimality guarantees
\begin{equation}
	\|(\bfr_\bfu)_{S_\bfu^{(t+1)}}\|_2 \ge \|(\bfr_\bfu)_{T^*}\|_2.
\end{equation}

Decomposing the proxy vector using $\bhSigma_{xy} = \rho \bfu \bfv^\top + \bfW$, we have $\bfr_\bfu = \rho \langle \bfv, \hat{\bfv}^{(t)} \rangle \bfu + \bfW \hat{\bfv}^{(t)}$. Let $\alpha = \langle \bfv, \hat{\bfv}^{(t)} \rangle$. Applying the triangle inequality to the selected support (LHS) and the reverse triangle inequality to the optimal support (RHS), we obtain
\begin{align}
	\nonumber        & \rho |\alpha| \|\bfu_{S_\bfu^{(t+1)}}\|_2 + \|(\bfW \hat{\bfv}^{(t)})_{S_\bfu^{(t+1)}}\|_2 \\
	\nonumber   \ge~ & \|(\bfr_\bfu)_{S_\bfu^{(t+1)}}\|_2                                                         \\
	\nonumber \ge~   & \|(\bfr_\bfu)_{T^*}\|_2                                                                    \\
	\ge~             & \rho |\alpha| \|\bfu_{T^*}\|_2 - \|(\bfW \hat{\bfv}^{(t)})_{T^*}\|_2.
\end{align}
Rearranging to isolate the targeted energy $\|\bfu_{S_\bfu^{(t+1)}}\|_2$:
\begin{align}
	\nonumber & \|\bfu_{S_\bfu^{(t+1)}}\|_2                                                                                                                     \\
	\ge~      & \|\bfu_{T^*}\|_2 - \frac{1}{\rho |\alpha|} \left( \|(\bfW \hat{\bfv}^{(t)})_{S_\bfu^{(t+1)}}\|_2 + \|(\bfW \hat{\bfv}^{(t)})_{T^*}\|_2 \right).
\end{align}

It remains to bound the projection of the noise terms. For any index set $K$ (representing either $S_\bfu^{(t+1)}$ or $T^*$), since $\hat{\bfv}^{(t)}$ is supported entirely on $S_\bfv^{(t)}$ and is a unit vector, we have
\begin{align}
	\nonumber  \|(\bfW \hat{\bfv}^{(t)})_K\|_2 & = \| \bfW_{K, S_\bfv^{(t)}} (\hat{\bfv}^{(t)})_{S_\bfv^{(t)}} \|_2 \\
	\nonumber                                  & \le \|\bfW_{K, S_\bfv^{(t)}}\|_2 \|\hat{\bfv}^{(t)}\|_2            \\
	                                           & = \|\bfW_{K, S_\bfv^{(t)}}\|_2.
\end{align}

The matrix $\bfW_{K, S_\bfv^{(t)}}$ constitutes a rectangular submatrix of size $|K| \times |S_\bfv^{(t)}|$. Because $|S_\bfu^{(t+1)}| = |T^*|$, both instantiations of the noise block have identical dimensions. Conditioned on the uniform event $\mathcal{E}$, the spectral norm of any such restricted block is bounded by
\begin{equation}
	\|\bfW_{K, S_\bfv^{(t)}}\|_2 \le C(1+\rho)\sqrt{\frac{(|S_\bfu^{(t+1)}| + |S_\bfv^{(t)}|)\log n}{m}}.
\end{equation}

Substituting this spectral bound back into the energy inequality, we arrive at the desired conclusion
\begin{align}
	     & \|\bfu_{S_\bfu^{(t+1)}}\|_2                                                                                                              \\
	\ge~ & \frac{1}{\sqrt{s_\bfu(|S_\bfu^{(t+1)}|)}} - \frac{2C(1+\rho)}{\rho |\alpha|} \sqrt{\frac{(|S_\bfu^{(t+1)}| + |S_\bfv^{(t)}|)\log n}{m}}.
\end{align}
The symmetric lower bound for $\|\bfv_{S_\bfv^{(t+1)}}\|_2$ follows an identical derivation.
\hfill \IEEEQEDopen